\documentclass[11pt,onecolumn]{IEEEtran}

\usepackage[mathscr]{eucal}
\usepackage[cmex10]{amsmath}
\usepackage{epsfig,epsf,psfrag}
\usepackage{amssymb,amsmath,amsthm,amsfonts,latexsym}
\usepackage{amsmath,graphicx,bm,xcolor,url}
\usepackage[caption=false]{subfig}
\usepackage{fixltx2e}
\usepackage{array}
\usepackage{verbatim}
\usepackage{bm}
\usepackage{algorithmic, cite}
\usepackage{algorithm}
\usepackage{verbatim}
\usepackage{textcomp}
\usepackage{mathrsfs}
\usepackage{epstopdf}

\usepackage{epsfig,subfig}
\usepackage{epstopdf}
\usepackage{hyperref}
\usepackage{overpic}

\usepackage{amsmath}
\usepackage{amssymb}
\usepackage{color}

\newcommand{\openone}{\leavevmode\hbox{\small1\normalsize\kern-.33em1}}
 \newcommand{\pobs}{p}
 \newcommand{\DDelta}{\Delta}

\newcommand{\hel}{\kappa}
\newcommand{\hhel}{\hat{\kappa}}
\newcommand{\heta}{\hat{\eta}}

\newtheorem{theorem}{Theorem}
\newtheorem{lemma}{Lemma}

\newcommand{\nn}{\nonumber}

\newcommand{\calE}{\mathcal{E}}

\newcommand{\calG}{\mathcal{G}}

\newcommand{\calL}{\mathcal{L}}
\newcommand{\calM}{\mathcal{M}}

\newcommand{\calR}{\mathcal{R}}
\newcommand{\calS}{\mathcal{S}}

\newcommand{\calX}{\mathcal{X}}

\newcommand{\bw}{\boldsymbol{w}}

\newcommand{\bY}{\boldsymbol{Y}}

\newcommand{\bZ}{\boldsymbol{Z}}

\newcommand{\rme}{\mathrm{e}}

\newcommand{\bbE}{\mathbb{E}}

\newcommand{\bbN}{\mathbb{N}}

\newcommand{\bbR}{\mathbb{R}}

\DeclareMathAlphabet{\mathbsf}{OT1}{cmss}{bx}{n}
\DeclareMathAlphabet{\mathssf}{OT1}{cmss}{m}{sl}

\DeclareSymbolFont{bsfletters}{OT1}{cmss}{bx}{n}
\DeclareSymbolFont{ssfletters}{OT1}{cmss}{m}{n}
\DeclareMathSymbol{\bsfGamma}{0}{bsfletters}{'000}
\DeclareMathSymbol{\ssfGamma}{0}{ssfletters}{'000}
\DeclareMathSymbol{\bsfDelta}{0}{bsfletters}{'001}
\DeclareMathSymbol{\ssfDelta}{0}{ssfletters}{'001}
\DeclareMathSymbol{\bsfTheta}{0}{bsfletters}{'002}
\DeclareMathSymbol{\ssfTheta}{0}{ssfletters}{'002}
\DeclareMathSymbol{\bsfLambda}{0}{bsfletters}{'003}
\DeclareMathSymbol{\ssfLambda}{0}{ssfletters}{'003}
\DeclareMathSymbol{\bsfXi}{0}{bsfletters}{'004}
\DeclareMathSymbol{\ssfXi}{0}{ssfletters}{'004}
\DeclareMathSymbol{\bsfPi}{0}{bsfletters}{'005}
\DeclareMathSymbol{\ssfPi}{0}{ssfletters}{'005}
\DeclareMathSymbol{\bsfSigma}{0}{bsfletters}{'006}
\DeclareMathSymbol{\ssfSigma}{0}{ssfletters}{'006}
\DeclareMathSymbol{\bsfUpsilon}{0}{bsfletters}{'007}
\DeclareMathSymbol{\ssfUpsilon}{0}{ssfletters}{'007}
\DeclareMathSymbol{\bsfPhi}{0}{bsfletters}{'010}
\DeclareMathSymbol{\ssfPhi}{0}{ssfletters}{'010}
\DeclareMathSymbol{\bsfPsi}{0}{bsfletters}{'011}
\DeclareMathSymbol{\ssfPsi}{0}{ssfletters}{'011}
\DeclareMathSymbol{\bsfOmega}{0}{bsfletters}{'012}
\DeclareMathSymbol{\ssfOmega}{0}{ssfletters}{'012}

\newcommand{\hatF}{\hat{F}}

\newcommand{\hatL}{\hat{L}}

\newcommand{\hatP}{\hat{P}}

\newcommand{\hatq}{\hat{q}}

\newcommand{\hatS}{\hat{S}}

\newcommand{\hatv}{\hat{v}}

\newcommand{\hatw}{\hat{w}}

\newcommand{\eps}{\varepsilon}

\DeclareMathOperator*{\argmax}{arg\,max}

\DeclareMathOperator{\var}{\mathsf{Var}}

\DeclareMathOperator{\poly}{poly}

\newcommand{\bone}{\mathbf{1}}
\newcommand{\etal}{\textit{et al.}}

\begin{document}

\flushbottom
\title{Adversarial Top-$K$ Ranking}

\author{Changho Suh  $\,\quad\,$
        Vincent Y.~F.~Tan$\,\quad \,$ Renbo Zhao
\thanks{C.\ Suh is  with the School of Electrical Engineering at Korea Advanced Institute of Science and Technology (email:\,chsuh@kaist.ac.kr).}
\thanks{V.~Y.~F.~Tan is with the Department of Electrical and Computer Engineering and the Department of Mathematics, National University of Singapore.
(email:\,vtan@nus.edu.sg).}
\thanks{R.~Zhao is with the Department of Electrical and Computer Engineering, National University of Singapore.
(email:\,elezren@nus.edu.sg).}   \thanks{C.~Suh is supported by a gift from Samsung.  V.~Y.~F.~Tan and R.~Zhao gratefully acknowledge financial support from the National University of Singapore  (NUS) under the  NUS Young Investigator Award R-263-000-B37-133.}  }

\IEEEpeerreviewmaketitle

\maketitle

\begin{abstract}
We study the top-$K$ ranking problem where the goal is to recover the set of top-$K$ ranked items out of a large collection of items based on partially revealed preferences. We consider an {\em adversarial crowdsourced} setting where there are two population sets, and  pairwise comparison samples drawn from one of the populations follow  the standard Bradley-Terry-Luce model (i.e., the chance of item $i$ beating item $j$ is proportional to the relative score of item $i$ to item $j$), while in the other population, the corresponding chance is inversely proportional to the relative score. When the relative size of the two populations   is known, we characterize the minimax limit on the sample size required  (up to a constant)  for reliably identifying the top-$K$ items, and
 demonstrate how it scales with the relative size.
Moreover, by leveraging a tensor decomposition method for disambiguating mixture distributions, we extend our result to the more realistic scenario in which the relative population size is unknown, thus establishing an upper bound on the fundamental limit of the sample size for recovering the top-$K$ set.
\end{abstract}


\section{Introduction}

Ranking is one of the fundamental problems that has proved crucial in a wide variety of
contexts---social choice~\cite{caplin1991aggregation,soufiani2014computing}, web search and information retrieval~\cite{Dwork2001}, recommendation systems~\cite{baltrunas2010group}, ranking individuals by group comparisons~\cite{huang08} and crowdsourcing~\cite{chen13:crowdsourcing}, to name a
few.
Due to its wide applicability, a large volume of work on ranking  has been done. The two main paradigms in the literature include spectral
ranking algorithms~\cite{Negahban2012,Dwork2001,brin1998anatomy} and maximum likelihood
estimation (MLE)~\cite{ford1957solution}.
While these ranking schemes yield reasonably good estimates which are faithful globally w.r.t.\ the latent preferences (i.e., low $\ell_2$ loss), it is not necessarily guaranteed that this results in optimal ranking accuracy. Accurate ranking
has more to do with how well the {\em ordering} of the estimates matches that of the true preferences (a discrete/combinatorial optimization problem), and less to do with how well we can estimate the true preferences (a continuous optimization problem).


In   applications, a ranking algorithm that outputs a total ordering of  all  the items is not only   overkill, but it also unnecessarily increases complexity. Often, we pay attention to only a {\em few} significant items. Thus, recent work such as that by Chen and Suh~\cite{chen-suh:topKranking} studied the top-$K$ identification task. Here, one aims to recover a correct set of top-ranked items only. This work characterized the minimax limit on the sample size required  (i.e., the sample complexity) for reliable top-$K$ ranking, assuming
 the Bradley-Terry-Luce (BTL) model~\cite{bradley1952rank,luce1959individual}.

While this result is concerned with practical issues, there are still limitations when modeling other realistic scenarios. The BTL model considered in~\cite{chen-suh:topKranking} assumes that the quality of pairwise comparison information which forms the basis of the model is the same across annotators. In reality  (e.g., crowdsourced settings), however, the quality of the information can vary significantly across different annotators. For instance, there may be a non-negligible fraction of spammers who provide answers in an \emph{adversarial} manner.
In the context of {\em adversarial  web search}~\cite{AdvWeb_10}, web contents can be maliciously manipulated by spammers for commercial, social, or political benefits in a robust manner.
Alternatively, there  may exist false information such as false voting in social networks and fake ratings in recommendation systems~\cite{AdvIR_02}.

As an initial effort to address this challenge, we investigate a so-called {\em adversarial BTL} model, which postulates the existence of two sets of populations---the \emph{faithful} and \emph{adversarial} populations, each of which has proportion $\eta$ and $1-\eta$ respectively. Specifically we consider a BTL-based pairwise comparison model  in which there exist latent variables indicating ground-truth preference scores of items.
In this model, it is assumed that comparison samples drawn from the faithful population follow the standard BTL model (the probability of item $i$ beating item $j$ is proportional to item $i$'s relative score to item $j$), and those of the adversarial population act in an ``opposite'' manner, i.e.,  the probability of   $i$ beating   $j$ is inversely proportional to the relative score. See Fig.~\ref{fig:model}.




\subsection{Main contributions} We seek to characterize the fundamental limits on the sample size required for top-$K$ ranking, and to develop computationally efficient ranking algorithms.  There are two main contributions in this paper.


Building upon \emph{RankCentrality}~\cite{Negahban2012} and \emph{SpectralMLE}~\cite{chen-suh:topKranking}, we develop a ranking algorithm to characterize the minimax limit required for top-$K$ ranking, up to constant factors, for the $\eta$-known scenario.
We also show the minimax optimality of our ranking scheme by proving a converse or impossibility result that applies to {\em any} ranking algorithm using information-theoretic methods.
As a result, we find that the sample complexity is inversely proportional to $(2 \eta-1)^2$, which suggests that less distinct the population sizes, the larger the sample complexity. 
We also demonstrate that our result recovers that of the $\eta=1$ case in~\cite{chen-suh:topKranking}, so the work contained herein is a strict generalization of that in \cite{chen-suh:topKranking}.

The second contribution is to establish an upper bound on the sample complexity for the more practically-relevant scenario where $\eta$ is unknown. A novel procedure based on tensor decomposition approaches in Jain-Oh~\cite{JO14} and Anandkumar \etal~\cite{AGHKT} is proposed  to first obtain  an estimate of the parameter $\eta$ that is in a   neighborhood of  $\eta$, i.e., we seek to obtain an $\eps$-globally optimal solution.
 This is usually not guaranteed by traditional iterative methods such as Expectation Maximization~\cite{Demp}.
Subsequently,   the estimate is then used in the ranking algorithm that assumes   knowledge of $\eta$. We demonstrate that this algorithm leads to an order-wise worse sample complexity relative to the $\eta$-known case.
Our theoretical analyses suggest  that the degradation is unavoidable if we employ this natural two-step procedure.


\subsection{Related work} The most relevant related works are those by  Chen and Suh~\cite{chen-suh:topKranking}, Negahban \etal~\cite{Negahban2012}, and Chen \etal~\cite{chen13:crowdsourcing}. Chen and Suh~\cite{chen-suh:topKranking} focused on top-$K$ identification under the standard BTL model, and derived an  $\ell_{\infty}$ error bound on preference scores which is intimately related to top-$K$ ranking accuracy. Negahban \etal~\cite{Negahban2012} considered the same comparison model and derived an $\ell_2$ error bound. A key distinction in our work is that we consider a different measurement model in which there are two population sets, although the $\ell_{\infty}$ and $\ell_2$ norm error analyses in~\cite{chen-suh:topKranking, Negahban2012} play crucial roles in determining the sample complexity.

The statistical model introduced by Chen \etal~\cite{chen13:crowdsourcing} attempts to represent crowdsourced settings and forms the basis of our adversarial comparison model. We note that no theoretical analysis of the sample complexity is available in~\cite{chen13:crowdsourcing} or other related works on crowdsourced rankings~\cite{yi13,ye13,kim14}. For example, Kim \etal~\cite{kim14} employed variational EM-based algorithms to estimate the latent scores;  {\em global} optimality guarantees for such algorithms are difficult to establish. Jain and Oh~\cite{JO14} developed a tensor decomposition method~\cite{AGHKT} for learning the parameters of a mixture model~\cite{Kearns94,Freund99,Feldman08} that includes our model as a special case. We specialize their model and  relevant  results to our setting for determining the accuracy of the estimated $\eta$. This allows us to establish an upper bound on the   sample complexity when  $\eta$ is unknown.

Recently, Shah and Wainwright~\cite{shah15} showed that a simple counting method~\cite{Borda1781} achieves order-wise optimal sample complexity for top-$K$ ranking under a general comparison model which includes, as special cases, a variety of parametric ranking models including the one under consideration in this paper (the BTL model). However, the authors made assumptions on the statistics of the pairwise comparisons which are different from that in our model. Hence, their result is not directly applicable to our setting.

\subsection{Notations}
We provide a brief summary of the notations used throughout the paper.
Let $[n]$ represent $\left\{ 1,2,\cdots,n\right\} $.
We denote by $\Vert \boldsymbol{w} \Vert $,
$\Vert \boldsymbol{w} \Vert _{1}$, $\Vert \boldsymbol{w} \Vert _{\infty}$
the $\ell_{2}$ norm, $\ell_{1}$ norm, and $\ell_{\infty}$ norm
of $\boldsymbol{w}$, respectively.
Additionally, for any two sequences $f(n)$ and $g(n)$, $f(n)\gtrsim g(n)$ or $f(n)=\Omega(g(n))$ mean
that there exists a (universal)  constant $c$ such that $f(n)\geq cg(n)$; $f(n)\lesssim g(n)$ or $f(n) = O(g(n))$
mean that there exists a constant $c$ such that $f(n)\leq cg(n)$; and
$f(n)\asymp g(n)$ or $f(n)=\Theta(g(n))$ mean that there exist constants $c_{1}$ and $c_{2}$
such that $c_{1}g(n)\leq f(n)\leq c_{2}g(n)$.  The notation $\poly(n)$ denotes a sequence in $O(n^c)$ for some $c>0$. 

\section{Problem Setup} \label{sec:model1}

We now describe the model which we will analyze subsequently. We assume that the observations used to learn the rankings are in the form of a limited number of pairwise comparisons over $n$ items.
In an attempt to reflect the adversarial crowdsourced setting of our interest in which there are two population sets---the \emph{faithful} and \emph{adversarial} sets---we adopt a comparison model introduced by Chen \etal~\cite{chen13:crowdsourcing}. This is a generalization of the BTL model~\cite{bradley1952rank, luce1959individual}. We delve into the details of the components of the model.




\begin{figure}[t]
\begin{center}
{\epsfig{figure=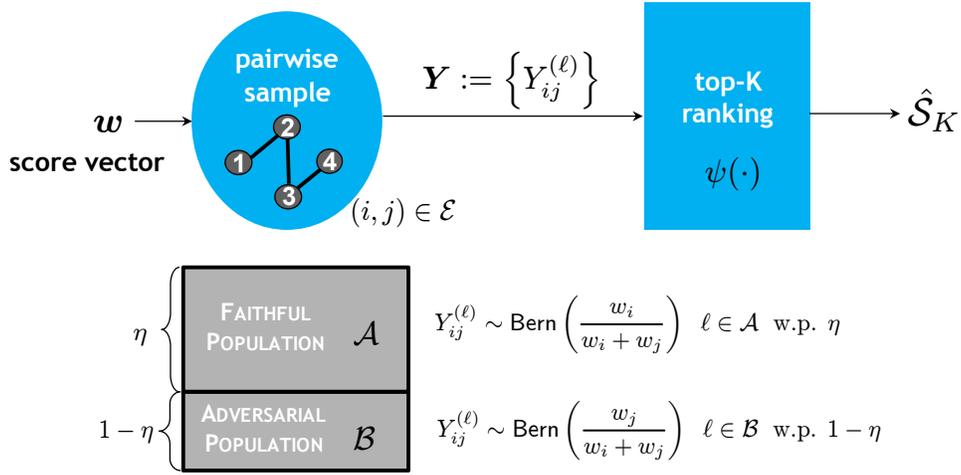, angle=0, width=0.7\textwidth}}
\end{center}
\vspace*{-0.1in}
\caption{Adversarial top-$K$ ranking given samples $\bY=\{Y_{ij}^{(\ell)}\}$ where $(i,j)\in\calE$ and $\calE$ is the edge set of an Erd\H{o}s-R\'enyi  random graph.}
\label{fig:model}
\vspace*{-0.1in}
\end{figure}

{\em Preference scores:} As in the standard BTL model,  this model postulates the existence of a ground-truth preference score vector $\bw = (w_1, w_2,\ldots, w_n)\in\bbR_+^n$. Each $w_i$ represents the underlying preference score of item $i$.
Without loss of generality, we assume that the scores are in non-increasing order:
\begin{equation}
w_1\ge w_2\ge \ldots\ge w_n> 0.   \label{eqn:decreasing}
\end{equation}
It is assumed that the dynamic range of the score vector is fixed irrespective of $n$:
\begin{equation}
w_i \in [w_{\min},w_{\max} ] , \qquad \forall\, i \in [n],
\end{equation}
for some positive constants $w_{\min}$ and $w_{\max} $. In fact, the case in which the ratio $\frac{w_{\max}}{w_{\min}}$ grows with $n$ can be readily translated into the  above  setting by first separating out those items with vanishing scores (e.g., via a simple voting method like Borda count~\cite{Borda1781,ammar-shah2011}).

{\em Comparison graph:} Let $\calG := ([n],\calE)$ be the comparison graph such that items $i$ and $j$ are compared by an annotator if the node pair $(i,j)$ belongs to the edge set $\calE$. We will assume throughout that the edge set $\calE$ is drawn in accordance to the Erd\H{o}s-R\'enyi (ER) model $\calG \sim {\cal G}_{n,p}$. That is node pair $(i,j)$ appears independently of any other node pair with an observation probability $\pobs \in(0,1)$.

{\em Pairwise comparisons:} For each edge $(i,j) \in\calE$, we observe $L$  comparisons between $i$ and $j$.
Each outcome, indexed by $\ell \in [L]$ and denoted by $Y_{ij}^{ (\ell)}$, is drawn from a mixture of Bernoulli distributions weighted by an unknown parameter $\eta \in (1/2,1]$.  The $\ell$-th observation of edge $(i,j)$  has distribution ${\sf Bern}(\frac{w_i}{w_i + w_j})$ with probability $\eta$ and  distribution $ {\sf Bern} (\frac{w_j}{w_i + w_j})$ with probability $1- \eta$. Hence,
\begin{equation}
Y_{ij}^{(\ell)} \sim {\sf Bern} \left( \eta\frac{w_i}{w_i+w_j }+ (1-\eta) \frac{w_j}{w_i+w_j } \right).
\label{eqn:mixture_mode}
\end{equation}
See Fig.~\ref{fig:model}. 
When $\eta=1/2$, all the observations are fair coin tosses. In this case, no information can be gleaned about the rankings.
Thus we exclude this degenerate setting from our study. The case of $\eta \in [0, 1/2)$ is equivalent to the ``mirrored'' case of $1-\eta \in (1/2, 1]$ where we flip $0$'s to $1$'s and $1$'s to $0$'s. So without loss of generality, we assume that $\eta \in (1/2, 1]$. We allow $\eta$ to depend on $n$.

Conditioned  on the graph $\calG$, the $Y_{ij}^{(\ell)}$'s are independent and identically distributed across all $\ell$'s, each according to the distribution of~\eqref{eqn:mixture_mode}.
The collection of sufficient statistics is
\begin{equation}
Y_{ij} := \frac{1}{L}\sum_{\ell=1}^L Y_{ij}^{ (\ell)} ,\qquad\forall\,  (i,j) \in\calE. \label{eq:defSuffStat}
\end{equation}
The {\em per-edge number of samples} $L$ is measure of the quality of the measurements.
We let  $\bY_i  :=\{ Y_{ij} \}_{j : (i,j)\in \calE}$, $\vec{Y}_{ij}:= \{Y_{ij}^{(\ell)} : \ell\in [L] \}$ and  $\bY := \{Y_{ij}\}_{   (i,j)\in\calE}$ be various statistics of the available data. 


{\em Performance metric:} We are interested in recovering the top-$K$ ranked items in the collection of $n$ items from the data  $\bY$.
We denote the true set of top-$K$ ranked items by $\calS_K$ which, by our ordering assumption, is the set $[K]$.
We would like to design a ranking scheme $\psi: \{0,1\}^{ |\calE| \times L}\to \binom{ [n ]}{K}$ that maps from the available measurements to a set of $K$ indices. 
Given a ranking scheme $\psi$, the performance metric we consider is the {\em probability of error}
\begin{equation}
P_{\rme} (\psi ) := \Pr \left[ \psi( \bY ) \ne \calS_K \right].
\end{equation}
We consider the fundamental \emph{admissible region} ${\cal R}_{\boldsymbol{w}}$ of $(\pobs,L)$ pairs in which top-$K$ ranking is feasible for a given $\boldsymbol{w}$, i.e., $P_{\rme}(\psi)$ can be arbitrarily small for large enough $n$. 
 In particular, we are interested in the {\em  sample complexity}
\begin{equation}
\label{eq:sample_complexity}
S_{\delta} := \inf_{p \in [0,1], L \in \mathbb{Z}^+ } \sup_{ \boldsymbol{a} \in \Omega_{ \delta} } \left\{ \binom{n}{2} p L : (p,L)\in\calR_{\boldsymbol{a}} \right\},
 \end{equation}
where $\Omega_{\delta}: = \{ \boldsymbol{a} \in \mathbb{R}^n: (a_{K} - a_{K+1})/a_{\max} \geq \delta \}$. Here we consider a \emph{minimax} scenario in which, given a score estimator,   nature can behave in an adversarial manner, and so she chooses the worst preference score vector that maximizes the probability of error under the constraint that the normalized score separation between the $K$-th and $(K+1)$-th items is at least $\delta$. Note that $\binom{n}{2}p$ is the expected number of edges of the ER graph so  $\binom{n}{2} p L$ is the expected number of pairwise samples  drawn from the model of our interest.

\section{Main Results} \label{sec:main_res}

As suggested in~\cite{chen-suh:topKranking}, a crucial parameter  for successful top-$K$ ranking is the separation between the two items near the decision boundary, 
\begin{equation}
\Delta_K :=\frac{w_K-w_{K+1}}{w_{\max}}. \label{eqn:Delta_K}
\end{equation}
The sample complexity depends on $\bw$ and $K$ only  through $\Delta_K$---more precisely, it decreases as  $\Delta_K$ increases. Our contribution  is to identify relationships between $\eta$ and the sample complexity when $\eta$ is known and unknown.  We will see that the sample complexity increases as $\Delta_K$ decreases. This is intuitively true as $\Delta_K$ captures how distinguishable the top-$K$ set is from the rest of the items.

We assume that the graph $\mathcal{G}$ is drawn from the ER
model ${\cal G}_{n,p}$  with edge appearance probability $\pobs$.
We require $\pobs$ to satisfy
\begin{equation}
\pobs>\frac{\log n}{n}.
\end{equation}
From random graph theory, this implies that the graph is connected with high probability. If the graph were not connected,   rankings cannot be inferred~\cite{ford1957solution}. 

We start by considering the $\eta$-known scenario in which key ingredients for ranking algorithms and analysis can be easily digested, as well as which forms the basis for the $\eta$-unknown setting.

\begin{theorem}[Known $\eta$]
\label{thm:etaknown}
Suppose that $\eta$ is known and ${\cal G} \sim {\cal G}_{n,p}$. Also assume that $L = O ( \poly (n))$ and $Lnp \geq \frac{c_0}{(2 \eta -1)^2} \log n$. Then with probability   $\ge 1 - c_1 n^{-c_2}$, the  set of top-$K$ set can be identified exactly provided
\begin{align}
L \geq c_3 \frac{ \log n}{ (2 \eta-1)^2 n p \Delta_K^2 }.
\end{align}
Conversely, for a fixed $\epsilon \in (0,\frac{1}{2})$, if
\begin{align}
L \leq c_4 \frac{ (1- \epsilon) \log n}{ (2 \eta-1)^2 n p \Delta_K^2 } \label{eqn:converse}
\end{align}
holds, then for any top-$K$ ranking scheme $\psi$, there exists a preference vector $\boldsymbol{w}$ with separation $\Delta_K$ such that $P_{\rme} (\psi) \geq \epsilon$. Here, and in the following, $c_i>0, i \in \{0,1,\ldots, 4\}$ are finite universal constants.
\end{theorem}
\begin{proof}
See Section~\ref{sec:ProofofTheorem1} for the algorithm and a sketch of the achievability proof (sufficiency). The proof of the converse (impossibility part) can be found in Section~\ref{sec:converse_etaknown}.
\end{proof}
This theorem asserts that the sample complexity scales as
\begin{equation}
S_{\Delta_K} \asymp   \frac{ n \log n }{ \left( 2\eta - 1 \right)^{{2}} \Delta_K^2 }.\label{eq:MinSampleComplexity}
\end{equation}
This result recovers that for the faithful scenario where $\eta=1$ in~\cite{chen-suh:topKranking}. When $\eta- \frac{1}{2}$ is uniformly bounded above $0$, we  achieve the same order-wise sample complexity. This suggests that the ranking performance is not substantially worsened if the sizes of the two populations are sufficiently distinct. For the challenging scenario in which $\eta \approx \frac{1}{2}$, the sample complexity depends on how $\eta - \frac{1}{2}$ scales with $n$.  Indeed, this dependence is quadratic. This theoretical result will be validated by experimental results in Section~\ref{sec:experiments}. Several other remarks are in order. 


{\em No computational barrier:} Our proposed algorithm  is based primarily upon two popular ranking algorithms: spectral methods and MLE, both of which enjoy nearly-linear time complexity in our ranking problem context. Hence, the information-theoretic limit promised by~\eqref{eq:MinSampleComplexity} can be achieved by a computationally efficient   algorithm.

{\em Implication of the minimax lower bound:} The minimax lower bound continues to hold when $\eta$ is unknown, since we can only do better for the $\eta$-known scenario, and hence the lower bound is also a lower bound in the $\eta$-unknown scenario.

{\em Another adversarial scenario:} Our results readily generalize to another adversarial scenario in which samples drawn from the adversarial population are {\em completely noisy}, i.e., they follow the distribution ${\sf Bern}(\frac{1}{2})$. With a slight modification of our proof techniques, one can easily verify that the sample complexity is on the order of  $\frac{ n \log n}{ \eta^2  \Delta_{K}^2}$ if $\eta$ is known. This will be evident after we describe the algorithm in Section~\ref{sec:ProofofTheorem1}.



\begin{theorem}[Unknown $\eta$]\label{thm:etaunknown}
Suppose that $\eta$ is unknown and $\mathcal{G} \sim {\cal G}_{n,p}$. Also assume that $L = O (  \poly(n))$ and $Lnp \geq \frac{c_0 }{(2 \eta -1)^4} \log^2 n$.
Then with probability   $\ge 1 - c_1 n^{-c_2}$, the top-$K$ set can be identified exactly provided
\begin{align}
L \geq c_3 \frac{ \log^2 n}{ (2 \eta-1)^4 n p \Delta_K^4 }.
\end{align}
\end{theorem}
\begin{proof}
 See Section~\ref{sec:ProofofTheorem2} for the key ideas in the proof.
\end{proof}
This theorem implies that the sample complexity satisfies
\begin{equation}
S_{\Delta_K} \lesssim  \frac{ n \log^2 n }{ \left( 2\eta - 1 \right)^{{4}} \Delta_K^{{4}} }.\label{eq:MinSampleComplexity_unknown}
\end{equation}
This bound is worse than~\eqref{eq:MinSampleComplexity}---the inverse dependence on $(2 \eta-1)^2 \Delta_{K}^2$ is now an inverse dependence on $(2 \eta -1)^4 \Delta_K^4$. This is because our algorithm involves estimating $\eta$, incurring some loss. Whether this loss is fundamentally unavoidable  (i.e., whether the algorithm is order-wise optimal or not) 
is open. See detailed discussions in Section \ref{sec:conclusion}. Moreover, since  the estimation of $\eta$ is based on tensor decompositions with polynomial-time complexity, our algorithm for the $\eta$-unknown case is also, in principle, computationally efficient.
Note that minimax lower bound in~\eqref{eq:MinSampleComplexity}   also serves as a lower bound in the $\eta$-unknown scenario.

%

%

\section{Algorithm and Achievability Proof of Theorem~\ref{thm:etaknown}}
\label{sec:ProofofTheorem1}

\begin{figure}[t]
\begin{center}
{\epsfig{figure=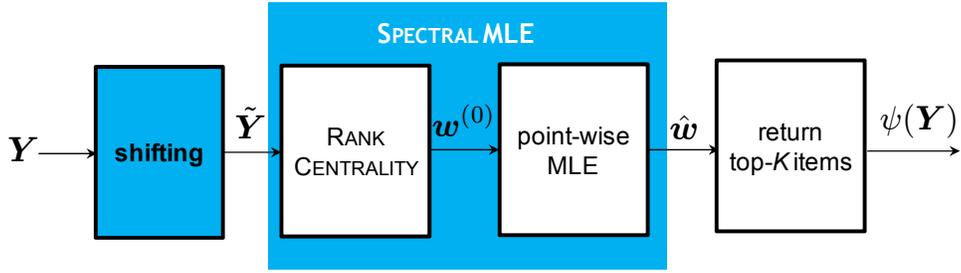, angle=0, width=0.7\textwidth}}
\end{center}
\vspace*{-0.1in}
\caption{
Ranking algorithm for the $\eta$-known scenario:
(1) shifting the empirical mean of pairwise measurements to get $\tilde{Y}_{ij} = \frac{ Y_{ij} - (1 -\eta) }{ 2 \eta -1 }$, which converges to $\frac{w_i}{w_i + w_j}$ as $L \rightarrow \infty$;
(2) performing SpectralMLE~\cite{chen-suh:topKranking} seeded by $\tilde{\boldsymbol{Y} }$ to obtain a score estimate $\hat{\boldsymbol{w}}$;
(3) return a ranking based on the estimate $\hat{\boldsymbol{w}}$. Our analysis reveals that the $\ell_{\infty}$ norm bound w.r.t. $\hat{\boldsymbol{w}}$ satisfies $\|  \hat{ \boldsymbol{w} } - \boldsymbol{w} \|_{\infty} \lesssim \frac{1}{2 \eta -1} \sqrt{ \frac{ \log n }{ npL} }$, which in turn ensures $P_e \rightarrow 0 $ under $\Delta_{K} \succsim \frac{1}{2 \eta -1} \sqrt{ \frac{ \log n }{ npL} }$.
}
\label{fig:etaknown}
\vspace*{-0.1in}
\end{figure}



\subsection{Algorithm Description}\label{sec:algo}

Inspired by the consistency between the preference scores $\bw$ and ranking under the BTL model, our scheme also adopts a two-step approach where $\bw$ is first estimated and then the top-$K$ set is returned. 

Recently a top-$K$ ranking algorithm \emph{SpectralMLE}~\cite{chen-suh:topKranking} has been developed for 
the faithful scenario and it  is shown to have  order-wise optimal sample complexity. The algorithm yields a small $\ell_{\infty}$ loss of the score vector $\bw$ which  ensures a small point-wise estimate error.
Establishing a key relationship between the $\ell_{\infty}$ norm error and top-$K$ ranking accuracy, Chen and Suh~\cite{chen-suh:topKranking} then identify an order-wise tight bound on the $\ell_{\infty}$ norm error required for top-$K$ ranking, thereby characterizing the sample complexity. Our ranking algorithm builds on SpectralMLE, which proceeds in two stages: (1) an appropriate initialization that concentrates around the ground truth in an $\ell_2$ sense, which can be obtained via spectral methods~\cite{Negahban2012,Dwork2001,brin1998anatomy}; (2) a sequence of $T$ iterative updates sharpening the estimates in a point-wise manner using MLE.


We observe that \emph{RankCentrality}~\cite{Negahban2012} can be employed as a spectral method in the first stage. In fact, RankCentrality exploits the fact that the empirical mean $Y_{ij}$ converges to the relative score $\frac{w_i}{w_i + w_j}$ as $L\to\infty$. This motivates the use of the empirical mean for constructing the transition probability from   $j$ to   $i$ of a Markov chain. 
Note that the detailed balance equation $\pi_i \frac{w_j}{w_i + w_j} = \pi_j \frac{w_i}{w_i + w_j}$ that holds  as  $L\to\infty$ will enforce that the stationary distribution of the Markov chain is identical to $\bw$ up to some constant scaling. Hence, the stationary distribution is expected to serve as a reasonably good global score estimate. However, in our problem setting where $\eta$ is not necessarily $1$, the empirical mean does not converge to the relative score, instead it behaves as
\begin{align}
Y_{ij} \,\, \stackrel{L \rightarrow \infty}{\longrightarrow }\,\,\eta \frac{w_i}{w_i + w_j} + (1-\eta) \frac{w_j}{w_i + w_j} .
\end{align}
Note, however, that the limit is linear in the desired relative score and   $\eta$, 
 implying that knowledge of $\eta$ leads to the relative score.  A natural idea then arises. We construct a shifted version of the empirical mean:
\begin{align}
\tilde{Y}_{ij} : = \frac{ Y_{ij} - (1- \eta) }{ 2\eta -1}  \,\,\stackrel{L \rightarrow \infty}{\longrightarrow } \,\, \frac{w_i}{w_i + w_j}, \label{eqn:def_tildeY}
\end{align}
and take this as an input to \emph{RankCentrality}. This   then forms a Markov chain that yields a stationary distribution that is proportional to $\bw$ as $L\to\infty$ and hence a good estimate of the ground-truth score vector when $L$ is large.
 This serves as a good initial estimate to the second stage of \emph{SpectralMLE} as it guarantees a   small point-wise error.

A formal and more detailed description of the procedure is summarized in Algorithm~\ref{Algorithm:SpectralMLE}. For completeness, we also include the procedure of \emph{RankCentrality} in Algorithm~\ref{Algorithm:RC}. Here we emphasize two distinctions w.r.t.\ the second stage of \emph{SpectralMLE}. First, the computation of the pointwise MLE w.r.t.\  say, item $i$, requires   knowledge of $\eta$:
\begin{align}
\label{eq:likelihoodfunction}
{\cal L}  \big( \tau, \boldsymbol{w}_{ \backslash i}^{(t)}; \boldsymbol{Y}_i \big)     =     
 \prod_{j: (i,j) \in {\cal E} } \Bigg[\bigg (  \eta \frac{ \tau }{ \tau  +  w_j^{(t)} }  +  (1 - \eta) \frac{ w_j^{(t)} }{ \tau  +  w_j^{(t)} }  \bigg)^{Y_{ij}} \bigg ( \eta \frac{ w_j^{(t)}  }{ \tau + w_j^{(t)} } + (1-\eta) \frac{ \tau }{ \tau + w_j^{(t)} }  \bigg)^{1- Y_{ij}}\Bigg].
\end{align}
Here, $  {\cal L}( \tau, \boldsymbol{w}_{ \backslash i}^{(t)}; \boldsymbol{Y}_i )$ is the  profile  likelihood of the preference score vector $[w_1^{(t)}, \cdots, w_{i-1}^{(t)}, \tau, w_{i+1}^{(t)}, \cdots, w_n^{(t)}]$ where $\boldsymbol{w}^{(t)}$ indicates the preference score estimate in the $t$-th iteration, $\boldsymbol{w}_{ \backslash i}^{(t)}$ denotes the score estimate excluding the $i$-th component, and $\boldsymbol{Y}_i$ is the data available at node $i$.
 The second difference is the use of  a different   threshold $\xi_t$ which incorporates the effect of $\eta$:
\begin{align}
\label{eq:xit}
   \xi_t  :=   \frac{c}{ 2 \eta  -  1 } \left \{ \sqrt{ \frac{ \log n }{ npL} }  +  \frac{1}{2^t}  \left(   \sqrt{ \frac{ \log n }{p L} }  -  \sqrt{ \frac{ \log n }{ npL} } \right) \right \},
\end{align}
where $c>0$ is a constant.
This threshold is used to decide whether $w_{i}^{(t+1)}$ should be  set to be the pointwise MLE $w_{i}^{{\sf mle}}$ in \eqref{eqn:ptwise_mle} (if $|w_{i}^{{\sf mle}}-w_{i}^{(t )}|>\xi_t$) or remains as $w_{i}^{(t)}$ (otherwise).
The design of $\xi_t$ is based on (1) the $\ell_\infty$ loss incurred in the first stage; and (2) a desirable $\ell_\infty$ loss that we intend to achieve at the end of the second stage. Since these two values are different, $\xi_t$ needs to be adapted accordingly. 
  Notice that the computation of $\xi_t$ requires   knowledge of $\eta$. The two modifications in \eqref{eq:likelihoodfunction} and \eqref{eq:xit} result in a more complicated analysis vis-\`a-vis Chen and Suh~\cite{chen-suh:topKranking}.

\begin{algorithm*}[t]
\caption{Adversarial top-$K$ ranking for the $\eta$-known scenario \label{alg:Nonconvex}}
\label{Algorithm:SpectralMLE}%
\begin{tabular}{>{\raggedright}p{1\textwidth}}
\textbf{Input}: The average comparison outcome $Y_{ij}$ for all
$(i,j)\in\mathcal{E}$; the score range $\left[w_{\min},w_{\max}\right]$.\vspace{0.7em}\tabularnewline
\textbf{Partition $\mathcal{E}$} randomly into two sets $\mathcal{E}^{\mathrm{init}}$
and $\mathcal{E}^{\mathrm{iter}}$ each containing $\frac{1}{2}\left|\mathcal{E}\right|$
edges. Denote by $\boldsymbol{Y}_{i}^{\mathrm{init}}$ (resp. $\boldsymbol{Y}_{i}^{\mathrm{iter}}$)
the components of $\boldsymbol{Y}_{i}$ obtained over $\mathcal{E}^{\mathrm{init}}$
(resp. $\mathcal{E}^{\mathrm{iter}}$). \vspace{0.7em}\tabularnewline
\textbf{Compute} the shifted version of the average comparison output: $\tilde{Y}_{ij} = \frac{ Y_{ij} - (1 - \eta) }{ 2 \eta -1 }$. Denote by $\tilde{\boldsymbol{Y}}_{i}^{\mathrm{init}}$ the components of $\tilde{\boldsymbol{Y}}_{i}$ obtained over $\mathcal{E}^{\mathrm{init}}$ \vspace{0.7em}\tabularnewline
\textbf{Initialize} $\boldsymbol{w}^{(0)}$ to be the estimate computed
by \emph{Rank Centrality} on $\tilde{\boldsymbol{Y}}_{i}^{\mathrm{init}}$
($1\leq i\leq n$).\vspace{0.7em}\tabularnewline
\textbf{Successive Refinement: for $t=0:T$ do}\tabularnewline
$\quad$1) Compute the coordinate-wise MLE \tabularnewline
$\qquad \qquad \qquad w_i^{\sf mle} \leftarrow \arg \max_{\tau} {\cal L} \left( \tau, \boldsymbol{w}_{ \backslash i}^{(t)}; \boldsymbol{Y}_i^{\sf iter}\right)$
\tabularnewline\hspace{.25in} where ${\calL}$ is the likelihood function defined in~\eqref{eq:likelihoodfunction}.\vspace{0.7em}\tabularnewline
$\quad$2)  For each $1\leq i\leq n$, set \tabularnewline
 $ \qquad \qquad \qquad    w_{i}^{(t+1)} \leftarrow
\left\{
  \begin{array}{ll}
   w_{i}^{\sf mle}, & \hbox{$|w_i^{\sf mle} - w_i^{(t)} | > \xi_t $;} \\
    w_{i}^{(t)}, & \hbox{else,}
  \end{array}
\right.
$ \tabularnewline
$\quad \;\;$ where $\xi_t$ is the replacement threshold defined in~\eqref{eq:xit}.
 \tabularnewline
\textbf{Output }the indices of the \textbf{$K$} largest components
of\textbf{ }$\boldsymbol{w}^{(T)}$.\tabularnewline
\end{tabular}
\end{algorithm*}

\begin{algorithm}[t]
\caption{Rank Centrality~\cite{Negahban2012}}
\label{Algorithm:RC}
\begin{tabular}{l}
\textbf{Input}: The shifted average comparison outcome $\tilde{Y}_{ij}$ for all
$(i,j)\in\mathcal{E}^{\mathrm{iter}}$.\vspace{0.7em}\tabularnewline
\textbf{Compute} the transition matrix $\hat{P}=[\hat{p}_{ij}]_{1\leq i,j\leq n}$ such that for $(i,j)\in\mathcal{E}^{\mathrm{iter}}$ \tabularnewline
$\qquad\qquad\qquad \hat{p}_{ij}=\begin{cases}
\frac{\tilde{Y}_{ji}}{d_{\max}},\quad & \text{if } i \neq j;\\
1-\frac{1}{d_{\max}}\sum_{k:(i,k)\in\mathcal{E}^{\mathrm{iter}}} \tilde{Y}_{ki},\quad & \text{if }i=j.
\end{cases}$ \vspace{0.3em}\tabularnewline
where $d_{\max}$ is the maximum out-degrees of vertices in $\mathcal{E}^{\mathrm{iter}}$. \vspace{0.7em}\tabularnewline
\textbf{Output} the stationary distribution of
$\hat{P}$.\tabularnewline
\end{tabular}
\end{algorithm}

\subsection{Achievability Proof of Theorem \ref{thm:etaknown}}\label{sec:prf_thm1}

Let $\hat{\boldsymbol{w}}$ be the final estimate $\boldsymbol{w}^{(T)}$ in the second stage. We carefully    analyze the $\ell_{\infty}$ loss of the $\bw$ vector, showing that under the conditions in Theorem~\ref{thm:etaknown} 
\begin{align}
\label{eq:linfty_norm_bound}
    \| \hat{\boldsymbol{w}} - \boldsymbol{w} \|_{\infty} \leq
    \frac{c_1}{2 \eta -1 } \sqrt{ \frac{ \log n }{ npL} }.
\end{align}
holds with probability exceeding $1- c_2 n^{-c_3}$.  This bound together with the following observation completes the proof. Observe that if $w_{K} - w_{K+1} \geq  \frac{c_4}{ 2\eta - 1}  \sqrt{ \frac{ \log n }{ n pL  }}$, then for a top-$K$ item $1 \leq i \leq K$ and a non-top-$K$ item $j \geq K+1$,
\begin{align}
\hat{w}_i - \hat{w}_j &\geq w_i - w_j - | w_i - \hat{w}_i | - | w_j - \hat{w}_j |  \label{eqn:arg1}\\
& \geq w_K - w_{K+1} - 2 \| \hat{\boldsymbol{w}} -  \boldsymbol{w} \|_{\infty} >0.\label{eqn:arg2}
\end{align}
This implies that our ranking algorithm outputs the top-$K$ ranked items as desired. Hence, as long as $ w_{K} - w_{K+1} \succsim  \frac{1}{ 2\eta - 1}  \sqrt{ \frac{ \log n }{ n pL  } }$ holds (coinciding with the claimed bound in Theorem~\ref{thm:etaknown}), we can guarantee perfect top-$K$ ranking, which completes the proof of Theorem~\ref{thm:etaknown}.

The remaining part is the proof of~\eqref{eq:linfty_norm_bound}. The proof builds upon the analysis made in~\cite{chen-suh:topKranking}, which demonstrates the relationship between $\frac{  \| \boldsymbol{w}^{(0)} - \boldsymbol{w} \|}{ \| \boldsymbol{w} \|}$ and $\| \boldsymbol{w}^{(T)} - \boldsymbol{w} \|_{\infty}$. We establish a new relationship for the arbitrary $\eta$ case, formally stated in the following lemma.  We will then use this to prove~\eqref{eq:linfty_norm_bound}.

\begin{lemma}
\label{lemma:l2vslinfty}
Fix $\delta,\xi>0$.  Consider $\hat{\boldsymbol{w}}^{\rm ub}$ such that it is independent of ${\cal G}$ and satisfies
\begin{align}
   \frac{\| \hat{\boldsymbol{w}}^{\rm ub} - \boldsymbol{w} \| }{ \| \boldsymbol{w} \| } \leq \delta\quad\mbox{and}\quad   \| \hat{\boldsymbol{w}}^{\rm ub} - \boldsymbol{w} \|_{\infty} \leq \xi  .\label{eqn:two_conditions}
\end{align}
Consider an estimate of the score vector $\hat{\boldsymbol{w}}$ such that $| \hat{w}_i - w_i | \leq | \hat{w}_i^{\rm ub} - w_i |$ for all $i\in [n]$. Let
\begin{equation}
 w_{i}^{\sf mle} :=  \argmax_{ \tau } {\cal L} (\tau,
\hat{\boldsymbol{w}}_{\backslash i}; \boldsymbol{Y}_{i} ). \label{eqn:ptwise_mle}
\end{equation}
Then, the pointwise error
\begin{align}
\label{eq:pointmlebound}
    |w_{i}^{\sf mle}  -   w_i | \leq c_0 \max
\left \{ \delta  +  \frac{ \log n}{ n p} \cdot \xi, \frac{c_1}{ 2 \eta -1} \sqrt{ \frac{ \log n}{ n p L} } \right \}
\end{align}
holds with probability at least $1 - c_2 n^{-c_3}$. 
\end{lemma}
\begin{proof}
The relationship in the faithful scenario $\eta=1$, which was proved in~\cite{chen-suh:topKranking}, means that the point-wise MLE $w_i^{\sf mle}$ is close to the ground truth $w_i$ in a component-wise manner, once an initial estimate $\hat{\boldsymbol{w}}$ is accurate enough. Unlike the faithful scenario, in our setting, we have (in general) noisier measurements $\boldsymbol{Y}_{i}$ due to the effect of $\eta$. Nonetheless  this lemma reveals that the relationship for the case of $\eta=1$ is almost the same as that for an arbitrary $\eta$ case only with a slight modification. This implies that a small point-wise loss is still guaranteed as long as we start from a reasonably good estimate. Here the only difference in the relationship is that the multiplication term of $\frac{1}{ 2 \eta -1}$ additionally applies in the upper bound of~\eqref{eq:pointmlebound}. See Appendix~\ref{app:Proofoflemmal2vslinfty} for the proof.
\end{proof}

Obviously the accuracy of the point-wise MLE reflected in the $\ell_{\infty}$  error depends crucially on an initial error $ \| \boldsymbol{w}^{(0)} - \boldsymbol{w} \|$. In fact, Lemma~\ref{lemma:l2vslinfty} leads to the claimed bound~\eqref{eq:linfty_norm_bound} once the initial estimation error is properly chosen as follows: 
\begin{equation}
\frac{  \| \boldsymbol{w}^{(0)} - \boldsymbol{w} \|}{ \| \boldsymbol{w} \|} \lesssim \frac{1}{ 2\eta-1} \sqrt{ \frac{\log n}{npL}}.
\end{equation}
 Here we demonstrate that the desired initial estimation error can indeed be achieved in our problem setting, formally stated in Lemma~\ref{lemma:l2-norm-bound-etaknown} (see below). On the other hand, adapting the analysis in~\cite{chen-suh:topKranking}, one can verify that with the replacement threshold $\xi_t$ defined in~\eqref{eq:xit}, the $\ell_2$ loss is monotonically decreasing in an order-wise sense, i.e., 
\begin{equation}
 \frac{  \| \boldsymbol{w}^{(t)} - \boldsymbol{w} \|}{ \| \boldsymbol{w} \|} \lesssim
\frac{  \| \boldsymbol{w}^{(0)} - \boldsymbol{w} \|}{ \| \boldsymbol{w} \|}. \label{eqn:mono_dec}
 \end{equation} 

We are now ready to prove~\eqref{eq:linfty_norm_bound} when  $L = O ( \poly  (n))$ and 
\begin{equation}
\frac{  \| \boldsymbol{w}^{(t)} - \boldsymbol{w} \|}{ \| \boldsymbol{w} \|} \asymp \delta \asymp \frac{1}{ 2\eta-1} \sqrt{ \frac{\log n}{npL}}. \label{eqn:choice_del}
\end{equation}
Lemma~\ref{lemma:l2vslinfty} asserts that in this regime, the point-wise MLE $\boldsymbol{w}^{\sf mle}$ is expected to satisfy
\begin{align}
\| \boldsymbol{w}^{\sf mle} -  \boldsymbol{w} \|_{\infty}  \lesssim  \frac{  \| \boldsymbol{w}^{(t)} - \boldsymbol{w} \|}{ \| \boldsymbol{w} \|}  +  \frac{ \log n }{ np} \| \boldsymbol{w}^{(t)}  -  \boldsymbol{w} \|_{\infty}.
\end{align}
Using the analysis in~\cite{chen-suh:topKranking}, one can show that the choice of $\xi_t$ in~\eqref{eq:xit}   enables us to detect outliers (where an estimation error is large) and drag down the corresponding point-wise error, thereby ensuring that $ \| \boldsymbol{w}^{(t+1)} - \boldsymbol{w} \|_{\infty} \asymp  \| \boldsymbol{w}^{\sf mle} - \boldsymbol{w} \|_{\infty}$. This together with the fact that
\begin{equation}
\frac{ \| \boldsymbol{w}^{(t)} - \boldsymbol{w} \|}{ \| \boldsymbol{w} \| } \lesssim \frac{ \| \boldsymbol{w}^{(0)} - \boldsymbol{w} \|}{ \| \boldsymbol{w} \| } \lesssim  \frac{1}{ 2\eta-1} \sqrt{ \frac{\log n}{npL}}
\end{equation} (see \eqref{eqn:choice_del} above and Lemma~\ref{lemma:l2-norm-bound-etaknown}) gives
\begin{align}
\label{eq:recursion2}
\| \boldsymbol{w}^{(t+1)}  -  \boldsymbol{w} \|_{\infty}
  \lesssim  \frac{1}{ 2\eta  -  1} \sqrt{\frac{ \log n}{ npL}}  +  \frac{ \log n }{ np} \| \boldsymbol{w}^{(t)}  -  \boldsymbol{w} \|_{\infty}.
\end{align}
A straightforward computation with this recursion yields~\eqref{eq:linfty_norm_bound}
if $\frac{\log n}{np}$ is sufficiently small (e.g., $p > \frac{ 2 \log n}{n}$) and $T$, the number of iterations in the second stage of \emph{SpectralMLE}, is sufficiently large (e.g., $T = O (\log n)$). 

\begin{lemma}
\label{lemma:l2-norm-bound-etaknown}Let $L   =O( \poly(n))$ and $L n p \geq \frac{c_0}{ (2 \eta -1)^2} \log n$. Let $\boldsymbol{w}^{(0)}$ be an initial estimate: an output of RankCentrality~\cite{Negahban2012} when seeded by  $\tilde{ \boldsymbol{Y} } := \{ \tilde{Y}_{ij} \}_{ (i,j) \in {\cal E} }$. Then,
\begin{align}
  \frac{  \| \boldsymbol{w} - { \boldsymbol{w}}^{(0)} \| }{ \| \boldsymbol{w} \| }
\leq \frac{c_1}{ 2\eta - 1}  \sqrt{ \frac{ \log n }{ n p L  } }\label{eqn:bound_l2_error0}
\end{align}
holds with probability exceeding $1 - c_2 n^{-c_3}$. 
\end{lemma}
\begin{proof}
Here we provide only a sketch of the proof, leaving details to Appendix~\ref{app:ProofofLemmaBoundDelta}. The proof builds upon the analysis structured by Lemma 2 in Negahban~\etal~\cite{Negahban2012}, which bounds the deviation of the Markov chain w.r.t.\ the transition matrix $\hat{P}$  after $t$ steps:
\begin{align}
\label{eq:deviationofMC}
   \frac{  \| \hat{p}_t - { \boldsymbol{w}} \| }{ \| \boldsymbol{w} \| }
\leq \rho^t   \frac{  \| \hat{p}_0 - { \boldsymbol{w}} \| }{ \| \boldsymbol{w} \| }  \sqrt{ \frac{w_{\max}}{w_{\min}} } + \frac{1}{ 1 - \rho} \| \DDelta \|  \sqrt{ \frac{w_{\max}}{w_{\min}} }
\end{align}
where $\hat{p}_t$ denotes the distribution w.r.t.\ $\hat{P}$ at time $t$ seeded by an arbitrary initial distribution $\hat{p}_0$, the matrix
$
\DDelta := \hat{P} - P,
$
indicates the fluctuation of the transition probability matrix\footnote{The notation $\DDelta=\hatP-P$, a matrix, should not be confused with the scalar normalized score separation $\Delta_K$, defined in \eqref{eqn:Delta_K}. } around its mean $P:= \mathbb{E} [\hat{P}]$, and $\rho:=\lambda_{\max} + \| \DDelta \| \sqrt{ \frac{w_{\max}}{w_{\min}} }$.
Here $\lambda_{\max} = \max \{ \lambda_2, - \lambda_n \}$ and $\lambda_i$ indicates the $i$-th eigenvalue of $P$.

Unlike the faithful scenario $\eta=1$,  in the arbitrary $\eta$ case,  the bound on $\| \DDelta \|$ depends on $\eta$:
\begin{align}
\label{eq:Delta_bound}
\| \DDelta \| \lesssim    \frac{1}{ 2\eta -1} \sqrt{ \frac{ \log n}{ npL } },
\end{align}
which will be proved in Lemma~\ref{app:ProofofLemmaBoundDelta} by using various concentration bounds (e.g., Hoeffding and Tropp~\cite{Tropp2011}). Adapting the analysis in~\cite{Negahban2012}, one can easily verify that $\rho < 1$ under one of  the conditions in Theorem~\ref{thm:etaknown} that $L n p \succsim \frac{ \log n}{ (2 \eta -1)^2 }$. Applying the bound on $\| \DDelta \|$ and $\rho <1 $ to~\eqref{eq:deviationofMC} gives the claimed bound,  which completes the proof.
\end{proof}

\section{Converse Proof of Theorem~\ref{thm:etaknown}}
\label{sec:converse_etaknown}

 As in Chen and Suh's work~\cite{chen-suh:topKranking}, by Fano's inequality, we see that it suffices for us to upper bound the mutual information between a set of appropriately chosen 
rankings $\calM$ of cardinality $M := \min\{K, n-K\}+1$.  More specifically, let $\sigma:[n]\to [n]$ represent a permutation over $[n]$. We also denote by $\sigma(i)$ and $\sigma([K])$ the
corresponding index of the $i$-th ranked item and the index set of all top-$K$ items, respectively. We subsequently impose a uniform prior over $\calM$ as follows: If $K<n/2$ then 
\begin{equation}
\Pr( \sigma( [K ] ) = \calS ) = \frac{1}{M} \quad\mbox{for}\quad \calS = \{2,\ldots, K\}\cup\{i\},\quad i = 1,K+1,\ldots, n \label{eqn:Klessn}
\end{equation}
and if $K\ge n/2$, then 
\begin{equation}
\Pr( \sigma( [K ] ) = \calS ) = \frac{1}{M} \quad\mbox{for}\quad \calS = \{1,\ldots, K+1\}\setminus\{i\},\quad i = 1,\ldots , K+1.\label{eqn:Kgen}
\end{equation}
In words, each alternative hypothesis is generated by swapping {\em only two} indices of the hypothesis (ranking) obeying $\sigma([K]) =[K]$. Clearly, the original minimax error probability is lower bounded by the corresponding error probability of this reduced ensemble.

Let the set of observations for the edge $(i,j) \in\calE$ be denoted as $\vec{Y}_{ij}:= \{Y_{ij}^{(\ell)} : \ell\in [L] \}$. We also find it convenient to introduce  an erased version of the observations $\bZ=\{\vec{Z}_{ij}: i,j\in [n]\}$ which is related to the true observations  $\bY:= \{\vec{Y}_{ij}: (i,j)\in \calE\}$ as follows,
\begin{equation}
\vec{Z}_{ij} = \left\{  \begin{array}{cc}
\vec{Y}_{ij}  & (i,j) \in\calE \\
\rme  & (i,j) \notin\calE
\end{array} \right. .
\end{equation}
Here $\rme$ is an {\em erasure} symbol.
Let $\sigma$, a chance variable, be a uniformly distributed ranking in $\calM$ (the ensemble of rankings created in~\eqref{eqn:Klessn}--\eqref{eqn:Kgen}). Let $P_{ \vec{Y}_{ij} |\sigma_j}$ be the distribution of the observations given that the ranking is $\sigma_j\in\calM$ where $j \in [M]$ and a similar notation is used for when   $\vec{Y}_{ij}$ is replaced by $\vec{Z}_{ij}$. Now, by the convexity of the relative entropy and the fact that the rankings are uniform, the mutual information can be bounded as
\begin{align}
I(\sigma;\bZ) & \le\frac{1}{M^2}\sum_{\sigma_1 ,\sigma_2\in\calM} D\left( P_{\bZ|\sigma_1} \big\| P_{\bZ|\sigma_2}  \right) \\
&  = \frac{1}{M^2}\sum_{\sigma_1 ,\sigma_2\in\calM} \sum_{i\ne j} D\left( P_{ \vec{Z}_{ij} |\sigma_1} \big\| P_{\vec{Z}_{ij} |\sigma_2}  \right) \\
&  = \frac{\pobs}{M^2}\sum_{\sigma_1 ,\sigma_2\in\calM} \sum_{i\ne j} D\left( P_{ \vec{Y}_{ij} |\sigma_1} \big\| P_{\vec{Y}_{ij} |\sigma_2}  \right) \\
&  = \frac{\pobs}{M^2}\sum_{\sigma_1 ,\sigma_2\in\calM} \sum_{i\ne j} \sum_{\ell=1}^L D\left( P_{ {Y}_{ij}^{(\ell)} |\sigma_1} \big\| P_{ {Y}_{ij}^{(\ell)} |\sigma_2}  \right). \label{eqn:bd_mi}
\end{align}
Assume that under ranking $\sigma_1$, the score vector is $\bw := (w_1,\ldots, w_n)$ and under ranking $\sigma_2$, the score vector is $\bw':=(w_{\pi(1)},\ldots, w_{\pi(n)})$  for some fixed permutation $\pi :[n]\to[n]$.  By using the statistical model described in Section~\ref{sec:model1}, we know that
\begin{equation}
 D\left( P_{ {Y}_{ij}^{(\ell)} |\sigma_1} \big\| P_{ {Y}_{ij}^{(\ell)} |\sigma_2}  \right) = D\left( \eta \frac{w_i}{w_i + w_j} + (1-\eta)\frac{w_j}{w_i + w_j}  \Big\| \eta\frac{w_{\pi(i)}}{w_{\pi(i)}+ w_{\pi(j)}} + (1-\eta) \frac{w_{\pi(j)}}{w_{\pi(i)}+ w_{\pi(j)}} \right) \label{eqn:bin_div}
\end{equation}
where $D(\alpha\| \beta): = \alpha\log\frac{\alpha}{\beta} + (1-\alpha)\log\frac{1-\alpha}{1-\beta}$ is the binary relative entropy. For brevity, write
\begin{equation}
a:=\frac{w_i}{w_i+w_j},\quad\mbox{and}\quad b:=\frac{w_{\pi(i)}}{w_{\pi(i)}+w_{\pi(j)}}.
\end{equation}
Furthermore, we note that the chi-squared divergence is an upper bound for the  relative entropy between two distributions $P=\{P_i\}_{i\in\calX}$ and $Q=\{Q_i\}_{i\in\calX}$ on the same (countable) alphabet $\calX$ (see e.g.~\cite[Lemma 6.3]{Csi06}), i.e.,
\begin{equation}
D( P \| Q) \le\chi^2 (P \| Q) :=\sum_{i\in\calX} \frac{(P_i-Q_i)^2}{Q_i}.\label{eqn:chi_bd}
\end{equation}
We also use the notation $\chi^2( \alpha\|\beta)$ to denote the  binary  chi-squared divergence similarly to the binary relative entropy.
Now, we may bound \eqref{eqn:bin_div} using the  following computation
\begin{align}
&D\left( \eta a+(1-\eta)(1-a)  \big\| \eta b+(1-\eta)(1-b)  \right) \nn\\*
&\le \chi^2 \left( \eta a+(1-\eta)(1-a)   \big\| \eta b+(1-\eta)(1-b )   \right)  \label{eqn:apply_chi_bd}\\
&=\frac{ ( 2\eta-1 )^2 (a-b)^2 }{  \big( (2\eta-1)b+(1-\eta)\big)\big( \eta-(2\eta-1)b\big) } \label{eqn:bd_div}
\end{align}
 Now
\begin{equation}
|a-b|\le \frac{w_K}{w_K+w_{K+1}}- \frac{w_{K+1}}{w_K+w_{K+1}}\le\frac{w_{\max}}{2w_{\min}}\Delta_K. \label{eqn:bd_a_minus_b}
\end{equation}
Hence, if we consider the case where $\eta=(1/2)^+$ (which is the regime of interest), uniting \eqref{eqn:bd_div} and  \eqref{eqn:bd_a_minus_b} we obtain
\begin{equation}
D\left( \eta a+(1-\eta)(1-a)  \big\| \eta b+(1-\eta)(1-b)  \right)\lesssim(2\eta-1)^2 \Delta_K^2.
\end{equation}
By construction of the hypotheses in \eqref{eqn:Klessn}--\eqref{eqn:Kgen}, conditional on any two distinct rankings $\sigma_1,\sigma_2\in\calM$, the distributions of $\vec{Y}_{ij}$  (namely $P_{ \vec{Y}_{ij} |\sigma_1}$ and $P_{ \vec{Y}_{ij} |\sigma_2}$) are
different over at most $2n$ locations    so
\begin{equation}
 \sum_{i\ne j} \sum_{l=1}^L D\left( P_{ {Y}_{ij}^{(\ell)} |\sigma_1} \big\| P_{ {Y}_{ij}^{(\ell)} |\sigma_2}  \right)\lesssim  n L(2\eta-1)^2 \Delta_K^2  .
\end{equation}
Thus, plugging this into the bound on the mutual information in  \eqref{eqn:bd_mi}, we obtain
\begin{equation}
I(\sigma;\bZ)  \lesssim  \pobs  n L(2\eta-1)^2 \Delta_K^2   . \label{eqn:bd_mi2}
\end{equation}
Plugging this into Fano's inequality, and using the fact that $M\le n/2$ (from $M=\min\{K,n-K\}+1$),  we obtain  
\begin{align}
P_\rme(\psi) & \ge 1-\frac{I(\sigma;\bZ) }{\log M}-\frac{1}{\log M}\\
& \ge 1-\frac{I(\sigma;\bZ) }{\log (n/2)}-\frac{1}{\log (n/2)}.
\end{align}
 Thus, if $S = \binom{n}{2}pL\le \frac{c_2(1-\epsilon)\log n }{(2\eta-1)^2\Delta_K^2}$ for some small enough but positive $c_2$,  we see that
\begin{equation}
 P_\rme(\psi)  \ge  \epsilon  .
\end{equation}
Since this is independent of the decoder $\psi$, the converse part  is proved. 

\section{Algorithm and Proof of Theorem~\ref{thm:etaunknown}}
\label{sec:ProofofTheorem2}

\subsection{Algorithm Description} \label{sec:algo2}
The proof of Theorem \ref{thm:etaunknown} follows by combining the results of Jain and Oh~\cite{JO14} with the analysis for the case when $\eta$ is known in Theorem \ref{thm:etaknown}. Jain and Oh were interested in disambiguating a mixture distribution from samples. This corresponds to our model in \eqref{eqn:mixture_mode}. They showed using tensor decomposition methods  that it is possible to find a globally optimal solution for the mixture weight $\eta$   using a computationally efficient   algorithm. They also provided an $\ell_2$ bound on the error of the distributions but  as mentioned, we are more interested in controlling the $\ell_\infty$ error so we estimate $\bw$ separately. The use of the $\ell_2$ bound in \cite{JO14} leads to a worse sample complexity for top-$K$ ranking.

Thus, in the first step, we will use the method in \cite{JO14} to estimate $\eta$ given the data samples (pairwise comparisons) $\bY$. The estimate is denoted as $\hat{\eta}$.   It turns out that one can specialize the result in \cite{JO14} with suitably parametrized ``distribution vectors''
\begin{align}
  \pi_0  := \begin{bmatrix}
   \hdots &  \displaystyle\frac{w_i}{w_i+w_j} &  \displaystyle\frac{w_j }{w_i+w_j} &   \displaystyle \frac{w_{i'}}{w_{i'}+w_{j'}} & \displaystyle \frac{w_{j'}}{w_{i'}+w_{j'}} &\hdots
\end{bmatrix}^T     \label{eqn:pi0}
\end{align}
and $\pi_1 := \bone_{2|\calE|} - \pi_0\in\bbR^{2|\calE|}$ and where in \eqref{eqn:pi0},  $(i,j)$ runs through all values in $\calE$.
 Hence, we are in fact applying \cite{JO14}  to a more restrictive setting where the two probability distributions represented by $\pi_0$ and $\pi_1$  are ``coupled'' but this does not preclude the application of the results in \cite{JO14}. In fact, this assumption makes the calculation of relevant parameters (in Lemma \ref{lem:scale}) easier. The relevant second and third moments are
\begin{align}
M_2 &:= \eta\pi_0\otimes\pi_0+(1- \eta) \pi_1\otimes \pi_1  ,  \label{eqn:defM2}\\
M_3 &:= \eta \pi_0\otimes\pi_0 \otimes\pi_0+(1- \eta) \pi_1\otimes \pi_1\otimes \pi_1,\label{eqn:defM3}
\end{align}
where $\pi_j \otimes \pi_j \in \bbR^{(2|\calE|)\times (2|\calE|)}$ is the outer product   and $\pi_j \otimes \pi_j \otimes \pi_j\in \bbR^{(2|\calE|)\times (2 |\calE|)\times (2 |\calE|)}$ is the $3$-fold tensor outer product. If one has the {\em exact} $M_2$ and $M_3$, we can obtain the mixture weight $\eta$ {\em exactly}.  The intuition as to why tensor   methods are applicable to problems involving latent variables has been well-documented (e.g.~\cite{AGHKT}). Essentially, the second- and third-moments contained in $M_2$ and $M_3$   provide sufficient statistics for identifying and hence estimating {\em all} the parameters of an appropriately-defined  model with latent variables (whereas second-order information contained in $M_2$ is, in general, not sufficient for reconstructing the parameters).   Thus,  the problem boils down to analyzing the precision of $\eta$ when we only have access to {\em empirical} versions of $M_2$ and $M_3$  formed from pairwise comparisons in $\calG$. As shown in Lemma \ref{lem:fidelity} to follow, there is a tradeoff between the sample size per edge $L$ and the quality of the estimate of $\eta$. Hence, this causes a degradation to the overall sample complexity reflected in Theorem \ref{thm:etaunknown}.

\begin{figure}[t]
\begin{center}
{\epsfig{figure=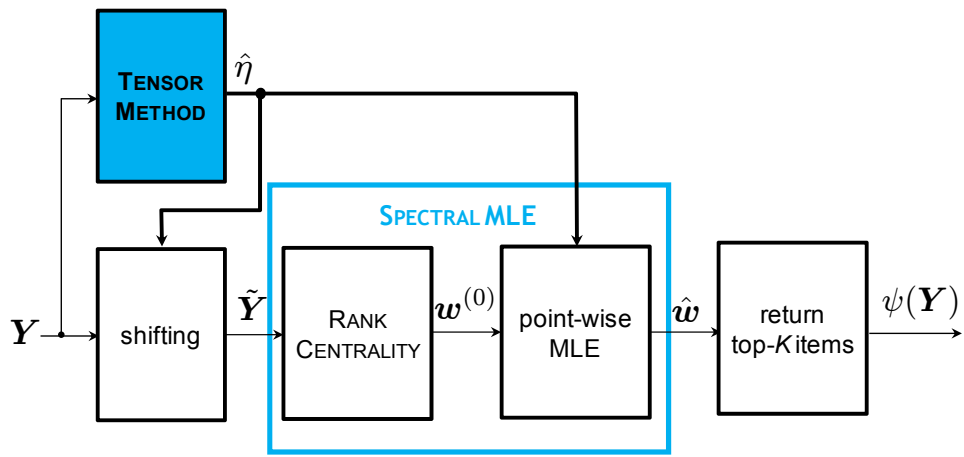, angle=0, width=0.7\textwidth}}
\end{center}
\vspace*{-0.1in}
\caption{
Ranking algorithm for the unknown $\eta$ scenario. The key distinction relative to the known $\eta$ case is that we estimate $\eta$ based on the tensor decomposition method \cite{JO14,AGHKT} and the estimate $\hat{\eta}$ is employed for shifting $\boldsymbol{Y}$ and performing the point-wise MLE. This method allows us to get $\|  \hat{ \boldsymbol{w} } - \boldsymbol{w} \|_{\infty} \lesssim \frac{1}{ 2 \eta -1 } \sqrt[4]{ \frac{ \log^2 n }{ npL} }$, which ensures that $P_e \rightarrow 0 $ under $\Delta_{K} \succsim \frac{1}{  2 \eta -1 } \sqrt[4]{ \frac{ \log^2 n }{ npL} }$.
}
\label{fig:etaunknown}
\vspace*{-0.1in}
\end{figure}

\begin{algorithm*}[t]
\caption{Estimating mixing coefficient $\eta$  \cite{JO14}}\label{alg:est_eta}
\begin{tabular}{>{\raggedright}p{1\textwidth}}
{\bf Input}: The collection of observed pairwise comparisons $\bY$\vspace{0.7em}\tabularnewline
{\bf Split} $\bY$ evenly into two   subsets  of samples $\bY^{(1)}$ and $\bY^{(2)}$\vspace{0.7em}\tabularnewline
{\bf Estimate} the second-order moment matrix $M_2$ in~\eqref{eqn:defM2} based on $\bY^{(1)}$ using Algorithm 2 (MatrixAltMin) in \cite{JO14}\vspace{0.7em}\tabularnewline
{\bf Estimate} a third-order statistic $G$ (defined in \cite[Theorem 1]{JO14}) based on $(M_2,M_3,\bY^{(2)})$ using Algorithm 3 (TensorLS) in \cite{JO14}\vspace{0.7em}\tabularnewline
{\bf Compute} the first eigenvalue $\lambda_1$ of $G$ using the {\em robust power method} in \cite{AGHKT}\vspace{0.7em}\tabularnewline
{\bf Return} the estimated mixing coefficient $\hat{\eta} = \lambda_1^{-2}$
\end{tabular}
\end{algorithm*}


In the second step, we plug the estimate $\heta$ into the  algorithm for the $\eta$-known case  by shifting the observations $\bY$ similarly to \eqref{eqn:def_tildeY} but with $\heta$ instead of $\eta$. See Fig.~\ref{fig:etaunknown}. However, here there are a couple of important distinctions relative to the case where $\eta$ is known exactly. First, the likelihood function $\calL(\cdot)$ in \eqref{eq:likelihoodfunction} needs to be modified since it is a function of $\eta$  in which now we only have its estimate $\heta$. Second,  since the guarantee on the $\ell_\infty$ loss of the preference score vector $\bw$ is different (and in fact worse), we need to design the threshold $\xi_t$ differently from \eqref{eq:xit}. We call the modified threshold $\hat{\xi}_t$, to be defined precisely in \eqref{eqn:hatxi_t}.

\subsection{Proof of Theorem \ref{thm:etaunknown}}
As in Section \ref{sec:prf_thm1}, the crux is to analyze the $\ell_\infty$ loss of the $\bw$ vector. We show that
\begin{equation}
\| \hat{\bw}-\bw\|_\infty\le  \frac{c_0}{  2\eta - 1 }  \sqrt[4]{ \frac{ \log^2 n }{ n p L  } } \label{eqn:linfty_unknown}
\end{equation}
holds with probability   $\ge 1-c_1 n^{-c_2}$. 
To guarantee accurate top-$K$ ranking, we then follow    the same argument as in~\eqref{eqn:arg1}--\eqref{eqn:arg2}. We   lower bound $\| \hat{\bw}-\bw\|_\infty$ in \eqref{eqn:linfty_unknown} by $\Delta_K$ and solve for $L$. Thus, it suffices to show~\eqref{eqn:linfty_unknown} under the conditions of Theorem~\ref{thm:etaunknown}.

The proof of \eqref{eqn:linfty_unknown} follows from several lemmata, two of which we present in this section. These are the analogues of Lemmas \ref{lemma:l2vslinfty} and \ref{lemma:l2-norm-bound-etaknown} for the $\eta$-known case.  Once we have these two lemmata, the strategy to proving \eqref{eqn:linfty_unknown} is almost the same as that in the $\eta$-known setting in Section \ref{sec:prf_thm1} so we omit the details.

The first lemma concerns the relationship between the normalized $\ell_2$ error and the $\ell_\infty$ error when we do not have access to the true  mixture weight $\eta$, but only an estimate of it given via Algorithm \ref{alg:est_eta}.


\begin{lemma}
\label{lemma:l2vslinfty_est}
Consider $\hat{\boldsymbol{w}}^{\rm ub}$ such that it is independent of ${\cal G}$ and satisfies \eqref{eqn:two_conditions}.
Consider $\hat{\boldsymbol{w}}$ such that $| \hat{w}_i - w_i | \leq | \hat{w}_i^{\rm ub} - w_i |$ for all  $i\in [n]$. Now define
\begin{align}
  w_{i}^{\sf mle} :=  \argmax_{ \tau } \hat{{\cal L}} (\tau,
\hat{\boldsymbol{w}}_{\backslash i}; \boldsymbol{Y}_{i} ), \label{eqn:max_lik_est}
\end{align}
where $ \hat{{\cal L}}(\cdot )$ is the surrogate likelihood (cf. \eqref{eq:likelihoodfunction}) constructed with $\heta$ in place of $\eta$.
Then, for all $i$, the same pointwise MLE bound in~\eqref{eq:pointmlebound} 
holds with probability $\ge 1- c_0 n^{-c_1}$.
\end{lemma}
\begin{proof}
The proof  parallels that of Lemma \ref{lemma:l2vslinfty} but is more technical. We analyze the fidelity of the estimate $\heta$ relative to $\eta$ as a function of  $L$ (Lemma~\ref{lem:fidelity}). This requires the specialization of Jain and Oh~\cite{JO14} to our setting. By proving several  continuity statements, we show that the estimated normalized log-likelihood (NLL) $\frac{1}{L}\log\hat{\calL}(\cdot) $ is {\em uniformly} close to the true NLL $\frac{1}{L}\log\calL(\cdot)$ w.h.p. This leads us to prove  \eqref{eq:pointmlebound}, which is the {\em same} as the $\eta$-known case.  The details are deferred to Appendix~\ref{app:lemma:l2vslinfty_est}.
\end{proof}

Similarly to the case where $\eta$ is known, we need to subsequently control the initial error $\| \bw^{(0)}-\bw\|$. For the $\eta$-known case, this is done in Lemma \ref{lemma:l2-norm-bound-etaknown} so the following lemma is an analogue of Lemma \ref{lemma:l2-norm-bound-etaknown}.

\begin{lemma}\label{lem:l2_error_unknown}
Assume the conditions of Theorem \ref{thm:etaunknown} hold.  Let $\boldsymbol{w}^{(0)}$ be an initial estimate, i.e., an output of RankCentrality  when seeded by $\tilde{\boldsymbol{Y}}$ which consists of the shifted observations with $\heta$ in place of  $\eta$ (cf.\ \eqref{eqn:def_tildeY}). Then,
\begin{align}
  \frac{  \| \boldsymbol{w} - { \boldsymbol{w}}^{(0)} \| }{ \| \boldsymbol{w} \| }
\leq \frac{c_0}{  2\eta - 1 }  \sqrt[4]{ \frac{ \log^2 n }{ n p L  } } \label{eqn:bound_l2_error}
\end{align}
holds with probability $\ge 1- c_1 n^{-c_2}$. 
\end{lemma}
\begin{proof}
See Section \ref{sec:prf_lem:l2_error_unknown} for a sketch of the proof and  Appendix~\ref{app:prf_upsilon_unknown} for a detailed calculation of an upper bound on the spectral norm of the fluctuation matrix, which is a key ingredient of the proof of Lemma \ref{lem:l2_error_unknown}.
\end{proof}
We remark that  \eqref{eqn:bound_l2_error} is worse than its $\eta$-known counterpart in \eqref{eqn:bound_l2_error0}. 
In particular, there is now a fourth root inverse dependence on $L$ (compared to a square root inverse dependence), which means we potentially need many more observations to drive the normalized $\ell_2$ error $ \frac{  \| \boldsymbol{w} - { \boldsymbol{w}}^{(0)} \| }{ \| \boldsymbol{w} \| }$ down to the same level.
This loss is present  because there is a penalty incurred in estimating $\eta$ via the tensor  decomposition approach, especially when $\eta$ is close to $1/2$. In the analysis, we need to control the Lipschitz constants of functions such as $t\mapsto \frac{1}{ 2t-1 }$ and $t\mapsto\frac{1-t}{2t-1}$ (see e.g.~\eqref{eqn:def_tildeY}). Such functions behave badly near $1/2$. In particular, the gradient diverges as $t\downarrow 1/2$.
We have endeavored to optimize   \eqref{eqn:bound_l2_error}  so that it is as tight as possible, at least using the proposed methods.

Using Lemmas~\ref{lemma:l2vslinfty_est} and~\ref{lem:l2_error_unknown} and invoking a similar argument as in the $\eta$-known scenario, we can now to prove~\eqref{eqn:linfty_unknown}. One key distinction here lies in the choice of the  threshold:
\begin{equation}
 \hat{\xi}_t  :=  \frac{c}{ 2 \hat{\eta} -1 } \left \{  \sqrt[4]{ \frac{ \log^2 n }{ npL} }  +  \frac{1}{2^t} \left(  \sqrt[4]{ \frac{ n \log^2 n }{p L} }  -  \sqrt[4]{ \frac{ \log^2 n }{ npL} }  \right)  \right \}    .
\label{eqn:hatxi_t}
\end{equation}
The rationale behind this choice, which is different from~\eqref{eq:xit}, is that it drives the initial $\ell_{\infty}$ loss (associated to the initial $\ell_2$ loss in Lemma~\ref{lem:l2_error_unknown}) to approach the desired $\ell_{\infty}$ loss   in~\eqref{eqn:linfty_unknown}. Taking this choice, which we optimized, and adapting the analysis in~\cite{chen-suh:topKranking} with Lemma~\ref{lemma:l2vslinfty_est}, one can verify that the $\ell_{\infty}$ loss is monotonically decreasing in an order-wise sense: $\frac{  \| \boldsymbol{w}^{(t)} - \boldsymbol{w} \|}{ \| \boldsymbol{w} \|} \lesssim
\frac{  \| \boldsymbol{w}^{(0)} - \boldsymbol{w} \|}{ \| \boldsymbol{w} \|}$ similarly to \eqref{eqn:mono_dec}. By applying Lemma~\ref{lemma:l2vslinfty_est} to the regime where $L = O (   \poly   (n))$ and 
\begin{equation}
\frac{  \| \boldsymbol{w}^{(t)} - \boldsymbol{w} \|}{ \| \boldsymbol{w} \|} \asymp \delta \asymp \frac{1}{ 2\eta-1} \sqrt[4]{ \frac{\log^2 n}{npL}},
\end{equation}
 we get 
\begin{align}
\| \boldsymbol{w}^{\sf mle}  -  \boldsymbol{w} \|_{\infty}  \lesssim  \frac{  \| \boldsymbol{w}^{(t)}  -  \boldsymbol{w} \|}{ \| \boldsymbol{w} \|}  +  \frac{ \log n }{ np} \| \boldsymbol{w}^{(t)}  -  \boldsymbol{w} \|_{\infty}.
\end{align}
As in the $\eta$-known setting, one can show that the replacement threshold $\hat{\xi}_t$ leads to $\| \boldsymbol{w}^{\sf mle} - \boldsymbol{w} \|_{\infty} \asymp \| \boldsymbol{w}^{(t)} - \boldsymbol{w} \|_{\infty}$. This together with Lemma~\ref{lem:l2_error_unknown} gives
\begin{align}
\| \boldsymbol{w}^{(t+1)}  -  \boldsymbol{w} \|_{\infty}
 \lesssim \frac{1}{ 2\eta  -  1 } \sqrt[4]{\frac{ \log^2 n}{ npL}}  +  \frac{ \log n }{ np} \| \boldsymbol{w}^{(t)}  -  \boldsymbol{w} \|_{\infty}.
\end{align}
A straightforward computation with this recursion yields the claimed bound as long as $\frac{\log n}{np}$ is sufficiently small (e.g., $p > \frac{ 2 \log n}{n}$) and $T$ is sufficiently large (e.g., $T = O (\log n)$). This completes the proof of~\eqref{eqn:linfty_unknown}.

\subsection{Proof Sketch of Lemma \ref{lem:l2_error_unknown}}\label{sec:prf_lem:l2_error_unknown}
The proof of Lemma \ref{lem:l2_error_unknown} relies on the fidelity of the estimate $\heta$ as a function of $L$ when we use the tensor decomposition approach  by Jain and Oh~\cite{JO14} on the problem at hand.

\begin{lemma}[Fidelity of $\eta$ estimate] \label{lem:fidelity}
If the   number of observations per observed node pair $L$  satisfies
\begin{equation}
L\succsim\frac{1}{\eps^2}\log\frac{n}{\delta}, \label{eqn:lower_bd_L} , 
\end{equation}
then the estimate $\heta$ is $\eps$-close to the true value $\eta$ with probability exceeding $1-\delta$.
\end{lemma}
\begin{proof}
The complete proof using Theorem \ref{thm:jo} and Lemma \ref{lem:scale} is provided in Section  \ref{prf:lem:fidelity}.
\end{proof}
We take $\delta=n^{-c_0}$ (for some constant $c_0>0$) in the sequel so \eqref{eqn:lower_bd_L} reduces to $L\succsim \frac{1}{\eps^2}\log n$.  A major contribution  in the present paper is to find a ``sweet spot'' for $\eps$; if it is chosen too small, $\|\hat{\bw}-\bw\|_\infty$ is reduced (improving the estimation error) but $L$ increases (worsening the overall sample complexity). Conversely, if $\eps$ is chosen to be too large, the requirement on $L$ in \eqref{eqn:lower_bd_L} is relaxed, but  $\|\hat{\bw}-\bw\|_\infty$ increases and hence, the overall sample complexity grows (worsens) eventually.  The estimate in \eqref{eqn:lower_bd_L} is  reminiscent of a Chernoff-Hoeffding bound estimate of the sample size per edge $L$ required to ensure that the   average of i.i.d.\ random variables is $\eps$-close to its mean with probability  $\ge 1-\delta$. However, the justification is more involved and requires  specializing Theorem \ref{thm:jo}  (to follow) to our setting.

Now, we denote the difference matrix $\DDelta:=\hatP-P$ in which $\hat{\eta}$ is used in place of $\eta$ as $\hat{\DDelta}$.  Now using Lemma~\ref{lem:fidelity}, several continuity arguments, and some concentration inequalities, we are able to  establish that
\begin{equation}
\|\hat{\DDelta}\| \lesssim \frac{1}{ 2\eta-1 } \sqrt[4]{ \frac{\log^2 n}{npL}}  \label{eqn:bd_upsilon}
\end{equation}
with probability $\ge 1-c_1 n^{-c_2}$.
The   inequality  \eqref{eqn:bd_upsilon} is proved in Appendix \ref{app:prf_upsilon_unknown}.
  Now similarly to the proof of Lemma~\ref{lemma:l2vslinfty}, $\rho<1$  under the conditions of Theorem~\ref{thm:etaunknown}. Applying  the bound on the spectral norm of $\|\hat{\DDelta}\|$ in~\eqref{eqn:bd_upsilon} to~\eqref{eq:deviationofMC} (which continues to hold in the $\eta$-unknown setting) completes the proof of Lemma \ref{lem:l2_error_unknown}.

\subsection{Proof of Lemma \ref{lem:fidelity} }\label{prf:lem:fidelity}
To prove Lemma \ref{lem:fidelity}, we  specialize the non-asymptotic bound  on the recovery of parameters in  a mixture model in~\cite{JO14} to our setting; cf.~\eqref{eqn:pi0}.  Before stating this, we introduce a few notations. Let the singular value decomposition of $M_2$, defined in \eqref{eqn:defM2}, be written as $M_2 = U \Sigma V^T$ where $\Sigma=\mathrm{diag}(\sigma_1(M_2),\sigma_2(M_2))$ and  $U\in\bbR^{(2 |\calE|) \times 2}$ the matrix consisting of the left-singular vectors, is further decomposed as
\begin{equation}
U = \begin{bmatrix}  ((U^{(1) } )^T &  (U^{(2) }  )^T&\ldots &  (U^{( |\calE| ) }  )^T\end{bmatrix}^T .
\label{eqn:defUk}
\end{equation}
Each submatrix $U^{(k) } \in \bbR^{2\times 2}$ where $k$ denotes a node pair.   We say that $M_2$ is {\em $\tilde{\mu}$-block-incoherent}  if the operator norms for all $|\calE|$ blocks of $U$, namely $U^{(k) }  $, are upper bounded as
\begin{equation}
\| U^{(k) } \|_2 \le\tilde{\mu} \sqrt{ \frac{2}{|\calE|}} ,\qquad  \forall\, k \in \calE . \label{eqn:block_inc}
 \end{equation}
 For  $M_2$, the smallest block-incoherent constant $\tilde{\mu}$   is known as the {\em block-incoherence of $M_2$}. We denote this as $\mu(M_2):=\inf\{ \tilde{\mu}: M_2 \mbox{ is }\tilde{\mu}\mbox{-block-incoherent}\}$. 

\begin{theorem}[Jain and Oh~\cite{JO14}] \label{thm:jo}
Fix any $\eps,\delta>0$.
There exists a polynomial-time algorithm in $|\calE|$,   $\frac{1}{\eps}$ and $\log\frac{1}{\delta}$  (Algorithm 1 in \cite{JO14})  such that if
\begin{equation}
|\calE| \succsim \frac{ \sigma_1(M_2)^{4.5}  \mu(M_2)   }{\sigma_2(M_2)^{4.5}   } \label{eqn:cardZ}
\end{equation}
 and for a large enough (per-edge) sample size  $L$   satisfying
\begin{equation}
 L \succsim \frac{\mu(M_2)\sigma_1(M_2)^{6}  |\calE|^3 }{\min\{\eta,1-\eta\} \sigma_2(M_2)^{9} } \cdot\frac{\log (n/\delta)}{\eps^2} , \label{eqn:lower_bdL}
\end{equation}
the estimate of the mixture weight $\hat{\eta}$ is $\eps$-close to the true mixture weight $\eta$
with probability exceeding $1-\delta$.
\end{theorem}

It remains to  estimate  the scalings of $\sigma_1(M_2), \sigma_2(M_2)$ and $\mu(M_2)$. These require   calculations based on $\pi_0,\pi_1$ and $M_2$  and are summarized in the following crucial lemma.
\begin{lemma}  \label{lem:scale}
For a fixed  sequence of  graphs with $|\calE|$ edges, 
\begin{align}
\sigma_i(M_2) & = \Theta(|\calE|) ,\quad  i = 1,2, \label{eqn:sigma_res1}\\
\mu (M_2) &=\Theta(1). \label{eqn:mu_M2}
\end{align}
\end{lemma}
\begin{proof}
The proof of this lemma can be found in Appendix~\ref{app:prf_scalings}.  It hinges on the fact that $\|\pi_0\|^2 = \|\pi_1\|^2$, as the populations have ``permuted''  preference scores.
\end{proof}
Now the proof of Lemma \ref{lem:fidelity} is immediate upon substituting \eqref{eqn:sigma_res1} into \eqref{eqn:cardZ}--\eqref{eqn:lower_bdL}. We then notice that $|\calE| = \Theta(n^2p) = \omega(1)$ with high probability   so \eqref{eqn:cardZ} is readily satisfied. Also $\frac{\mu(M_2)\sigma_1(M_2)^{6}  |\calE|^3 }{\min\{\eta,1-\eta\}\sigma_2(M_2)^{9} } =\Theta(1)$ so we recover~\eqref{eqn:lower_bd_L} as desired.

\section{Experimental Results}
\label{sec:experiments}

For the case where $\eta$ is known, a number of   experiments on synthetic data were conducted to validate Theorem \ref{thm:etaknown}. We first state parameter settings common to all experiments. The total number of items is $n=1000$ and the number of ranked items $K=10$. In the pointwise MLE step in Algorithm \ref{Algorithm:SpectralMLE}, we set the number of iterations $T=\lceil\log n\rceil$ and $c=1$ in the formula for the threshold $\xi_t$ in \eqref{eq:xit}. The observation probability of each edge of the Erd\H{o}s-R\'{e}nyi graph is $p = \frac{6\log n}{n}$. The latent scores are uniformly generated from the dynamic range $[0.5,1]$. Each (empirical) success rate is averaged over $1000$ Monte Carlo trials.

\begin{figure}
\subfloat[]{\includegraphics[width=0.45\columnwidth]{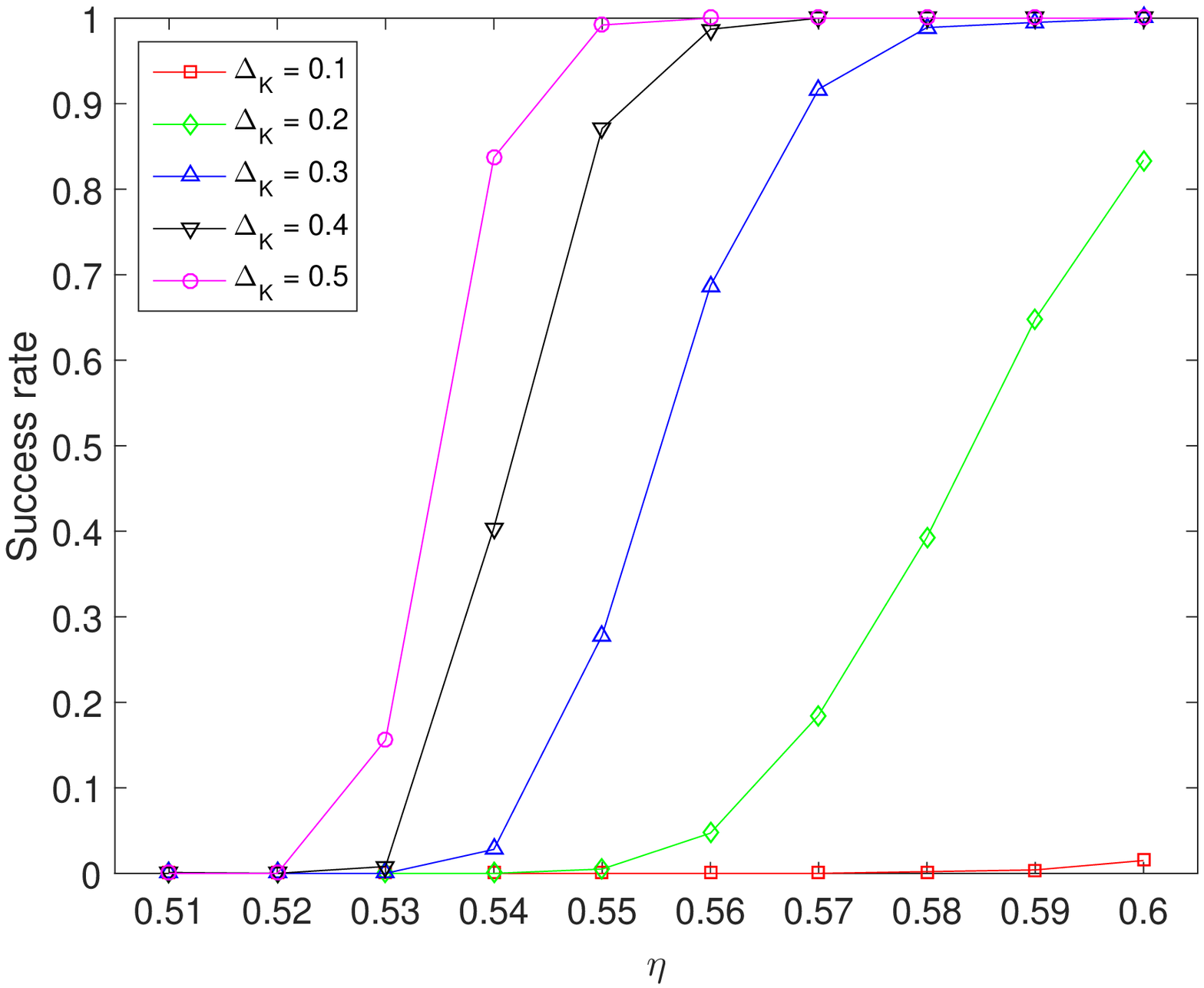}} \hspace{.1in}
\subfloat[]{\includegraphics[width=0.45\columnwidth]{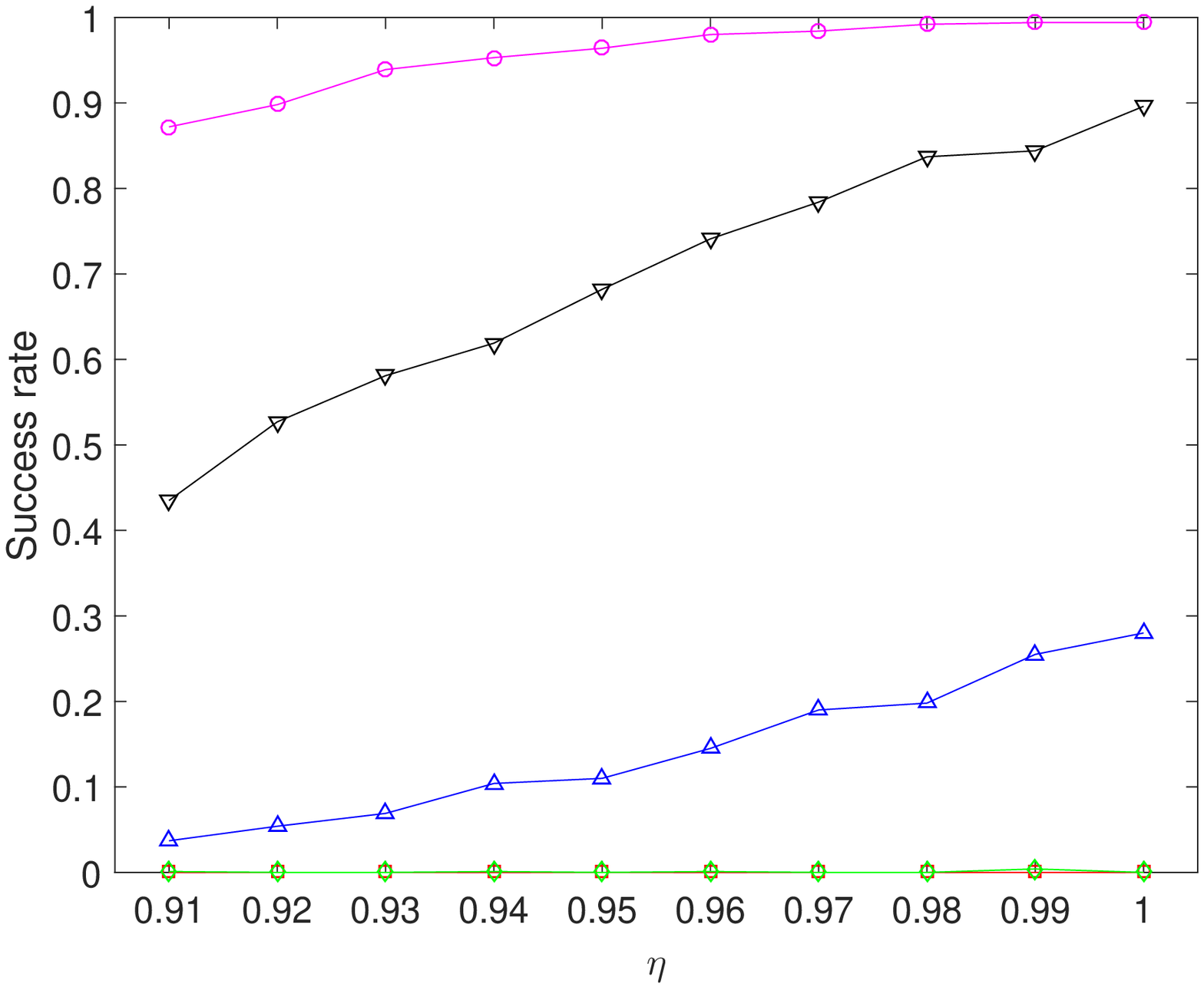}}\\
\centering
\caption{Success rates across $\eta$ for (a) $\eta$ close to $1/2$ and  (b) $\eta$ close to $1$.}\label{fig:succ_eta}
\end{figure}

We first examine the relations between success rates and $\eta$ for various values of the normalized separation of the scores $\Delta_K\in\{0.1,0.2,\ldots,0.5\}$. Here we consider  two different scenarios, one being  such that $\eta$ is close to $1/2$ and the other being such that  $\eta$  is close to $1$. We set the number of samples per edge, $L=1000$ for the first case and $L=10$ for the second. This is because when $\eta$ is small, more data samples are needed to achieve non-negligible success rates. The results for these two scenarios are shown in Figs.~\ref{fig:succ_eta}(a) and \ref{fig:succ_eta}(b) respectively. For both cases, when $L$ is fixed, we observe as $\eta$ increases, the success rates increase accordingly. However, the effect of $\eta$ on success rates is more prominent when $\eta$ is close to $1/2$. This is in accordance to  \eqref{eq:MinSampleComplexity} in Theorem \ref{thm:etaknown} since $1/(2\eta-1)^2$  has sharp decrease (as $\eta$ increases) near $1/2$ and a gentler decrease near $1$. Also, success rates increase when $\Delta_K$ increases. This again corroborates \eqref{eq:MinSampleComplexity} 
which says that the sample complexity is proportional to $1/\Delta_K^2$.

\begin{figure}
\center
\includegraphics[width=0.5\textwidth]{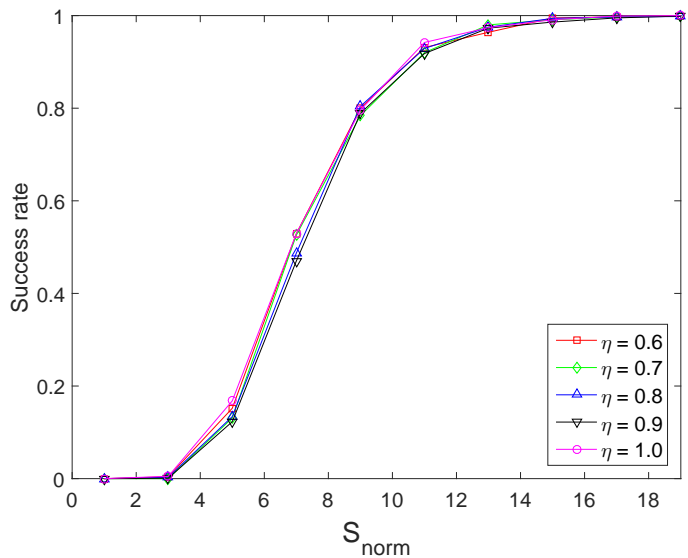}
\caption{Success rates across normalized sample size $S_{\mathrm{norm}}$.}
\label{fig:succ_L}
\end{figure}

Next we examine the relations between success rates and normalized sample size
\begin{equation}
S_{\mathrm{norm}} := \frac{S_{\Delta_K}}{ (n\log n)/ [(2\eta-1)^2\Delta_K^2]},
\end{equation}
for $\eta\in\{0.6,0.7,\ldots,1\}$. We fix $\Delta_K=0.4$ in this case. The results are shown in Fig.~\ref{fig:succ_L}. We observe the relations between success rates and $S_{\mathrm{norm}}$ are almost the same for all $\eta$'s so the implied constant factor in $\asymp$ notation in~\eqref{eq:MinSampleComplexity} depends very weakly on $\eta$ (if at all).  

Finally we numerically  examine the relation  between the   sample complexity   and $\eta$. 
We fix  $\Delta_K=0.4$ and focus on the regime where $\eta$ is close to $1/2$. 
For each $\eta$, we use the bisection method to approximately find the minimum sample size per edge $\hatL$ that achieves a high success rate $q_{\mathrm{th}}=0.99$. 
Specifically, the bisection procedure terminates when the empirical success rate $\hatq$ corresponding to $\hatL$ satisfies $|\hatq-q_{\mathrm{th}}|<\epsilon$, where $\epsilon$ is set to $5\times 10^{-3}$. We repeat such a  procedure $10$ times to get an average result $\hatL_{\mathrm{ave}}$. We also compute the resulting  standard deviation and observe that it is small across the $10$ independent runs.  Define the expected minimum total sample size
\begin{equation}
\hatS := \binom{n}{2}p\hatL_{\mathrm{ave}}.
\end{equation}
 To illustrate the  explicit dependence of $\hatS$ on $\eta$, we further normalize $\hatS$ to
\begin{equation}
 \hatS_{\mathrm{norm}} := \frac{\hatS}{(n\log n)/\Delta_K^2},
 \end{equation}
thus  isolating the dependence of  minimum total sample size  on $\eta$ only. We then fit a curve $C/(2\eta-1)^2$  to $\hatS_{\mathrm{norm}}$, where $C$ is chosen to best fit the points by optimizing  a least-squares-like objective function. The empirical  results (mean and one standard deviation) together with the fitted curve are shown in Fig.~\ref{fig:L_min_eta}. We observe $\hatS_{\mathrm{norm}}$ depends on $\eta$ via $1/(2\eta-1)^2$ almost perfectly  up to a constant. This corroborates our theoretical result in \eqref{eq:MinSampleComplexity}, i.e.,  the reciprocal  dependence of the sample complexity  on $(2\eta-1)^2$.

\begin{figure}
\label{fig:validate_theory}
\center
\includegraphics[width=0.5\textwidth]{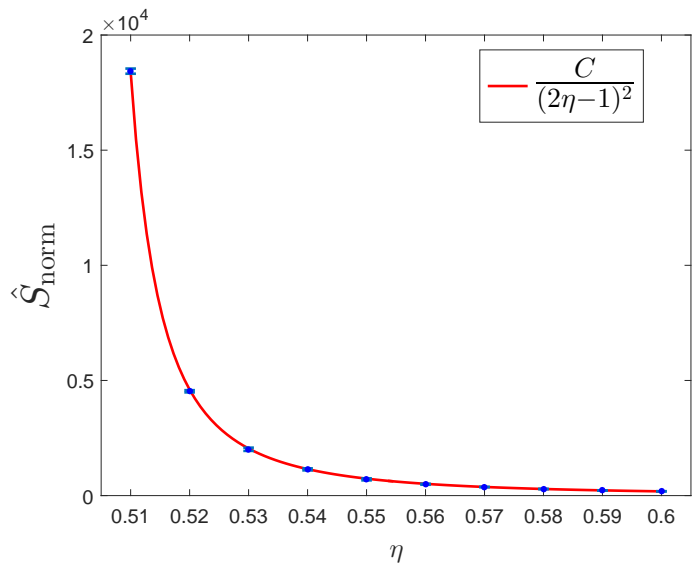}
\caption{Normalized empirical sample size $\hatS_{\mathrm{norm}}$ for $\eta$ close to $1/2$.}\label{fig:L_min_eta}
\end{figure} 

For the case where $\eta$ is not known, the storage costs turn out to be prohibitive even for a moderate number of items $n$. Hence, we leave the implementation of the algorithm for the $\eta$-unknown case to future work. It is likely that one may need to formulate the ranking problem in an online manner~\cite{weng11} or  resort to online methods for performing tensor decompositions~\cite{Cichocki_13,Ge_15,Huang_15}.
\section{Conclusion and Further Work} \label{sec:conclusion}
In this paper, we have provided an analytical  framework for addressing the problem of recovering the  top-$K$ ranked items in an adversarial crowdsourced setting. We considered two scenarios. First, the proportion of adversaries $1-\eta$ is known and the second, more challenging scenario, is when this parameter is unknown.  For the first scenario, we adapted the \emph{SpectralMLE}~\cite{chen-suh:topKranking} and \emph{RankCentrality}~\cite{Negahban2012} algorithms to provide an order-wise optimal sample complexity bound for the total number of measurements for recovering the exact top-$K$ set. These results were verified numerically and the dependence of the sample complexity on the reciprocal of $(2\eta-1)^2$ was corroborated. For the second scenario, we adapted Jain and Oh's   global optimality result for disambiguating a mixture of discrete distributions~\cite{JO14} to first learn $\eta$. Subsequently, we plugged this (inexact) estimate into the known-$\eta$ algorithm and utilized  a sequence of  continuity arguments to obtain an upper bound on the sample complexity. This bound is   order-wise worse than the case where $\eta$ is known, showing that the error induced by the estimation  of the mixture parameter dominates the overall procedure.

A few natural questions result from our analyses.

\begin{enumerate}
\item Can we close the gap in the sample complexities between the $\eta$-known and $\eta$-unknown scenarios? This seems challenging given that (i) threshold $\hat{\xi}_t$ in \eqref{eqn:hatxi_t} must not be dependent on parameters that are assumed to be unknown such as the weight separation $\Delta_K$ and (ii) the fundamental difficulty of obtaining a globally optimal solution for the fraction of adversaries from samples that are drawn from a mixture distribution. Thus, we conjecture that if we adopt a two-step approach---first estimate $\eta$, then plug this estimate into the $\eta$-known algorithm---such a loss is unavoidable. This is because the fidelity of the estimate of $\eta$ in Lemma \ref{lem:fidelity} is natural (cf.\ Chernoff-Hoeffding bound) and does not seem to be order-wise improvable. Thus, we opine that a new class of algorithms, avoiding the explicit estimation of $\eta$, needs to be developed to improve the overall sample complexity performance. 
\item If closing the gap is difficult, can we hope to derive a converse or impossibility result, explicitly taking  into account the fact that $\eta$ is unknown? Our current converse result assumes that $\eta$ is known, which may be too optimistic for the unknown setting.
\item The tensor decomposition method~\cite{AGHKT,JO14}, while being polynomial time in its parameters, incurs high storage costs. Hence, in practice, it implementation to yield meaningful estimates of $\eta$ is challenging. There has been some recent progress on large-scale scalable tensor decomposition algorithms  in \cite{Cichocki_13,Ge_15,Huang_15}. In these works, the authors aim to avoid storing and manipulating large tensors directly. However,  since implementation is not the focus of the present work, we leave this to future work. 
\item  Recent work by Shah and Wainwright~\cite{shah15} has shown that simple counting methods for certain observation models (including the BTL model) achieve   order-wise optimal sample complexities. 
In the observation model considered therein,    for each pair of items $i$ and $j$, there is a random number of observations $R_{ij}$  that follows a binomial distribution with parameters $L\in\bbN$ and probability of success $p \in (0,1)$. 
Notice that  the observation model in \cite{shah15} differs from ours. 
\item Lastly, it would be interesting to consider other choice models (e.g., the Plackett-Luce model \cite{Plackett1975} studied in  \cite{hajek2014minimax} and \cite{maystre15})   as well as other comparison graphs not limited to the ER graph, as the comparison graph structure affects the sample complexity significantly, as suggested in~\cite[Theorem 1]{Negahban2012}.
\end{enumerate} 
\numberwithin{equation}{section}
\renewcommand{\theequation}{\thesection.\arabic{equation}}

\appendices


\section{Proof of Lemma~\ref{lemma:l2vslinfty}}
\label{app:Proofoflemmal2vslinfty}

For ease of presentation, we will henceforth assume that $w_{\max}=1$ since this simply amounts to a rescaling of all the preference scores. 

To prove the lemma, it suffices to show that if $\tau$ satisfies $   |\tau -  w_i | \succsim \max
\left \{ \delta + \frac{ \log n}{ n p} \cdot \xi,\frac{1}{ 2\eta -1}   \sqrt{ \frac{ \log n}{ n p L} } \right \}
$, then the corresponding likelihood function cannot be the point-wise MLE:
\begin{align}
\label{eq:Ltau_Lwi}
 {\cal L} (\tau,
\hat{\boldsymbol{w}}_{\backslash i}; \boldsymbol{y}_{i} ) <
   {\cal L} (w_i,
\hat{\boldsymbol{w}}_{\backslash i}; \boldsymbol{y}_{i} ).
\end{align}

We start by evaluating the likelihood function w.r.t. the ground-truth score vector:
\begin{align}
  \ell^* (\tau)  &:= \frac{1}{L} \log {\cal L} (\tau, \boldsymbol{w}_{\backslash i}; \boldsymbol{Y}_{i} ) \label{eqn:likelihood_fn0}\\
  & =   \sum_{j : (i,j) \in {\cal E}} \left \{
 Y_{ij} \log \left ( \eta \frac{ \tau }{ \tau + w_j } + (1- \eta) \frac{w_j}{ \tau + w_j} \right ) + ( 1 -  Y_{ij} ) \log \left( \eta \frac{ w_j }{ \tau + w_j } + (1- \eta) \frac{\tau}{ \tau + w_j} \right)
\right \}. 
\label{eqn:likelihood_fn}
\end{align}
The likelihood loss w.r.t.\ $w_i$ and $\tau$ is then computed as:
\begin{align}
   \ell^* (w_i) - \ell^* (\tau) =  \sum_{j : (i,j) \in {\cal E}} \left \{
 Y_{ij}  \log \left ( \frac{ \eta \frac{ w_i }{ w_i + w_j } + (1- \eta) \frac{w_j}{ w_i + w_j} }{    \eta \frac{ \tau }{ \tau + w_j } + (1- \eta) \frac{w_j}{ \tau + w_j}  } \right ) +
(  1 - Y_{ij} ) \log \left( \frac{  \eta \frac{ w_j }{ w_i + w_j } + (1- \eta) \frac{w_i}{ w_i + w_j} }{    \eta \frac{ w_j }{ \tau + w_j } + (1- \eta) \frac{\tau}{ \tau + w_j}   } \right)
\right \}.
\end{align}
Taking expectation w.r.t.\ $\boldsymbol{Y}_i$ conditional on ${\cal G}$, we get:
\begin{align}
 \mathbb{E} \left[   \ell^* (w_i) - \ell^* (\tau) | {\cal G} \right] &=   \sum_{j : (i,j) \in {\cal E}}
D \left( \eta \frac{ w_i }{ w_i + w_j } + (1- \eta) \frac{w_j}{ w_i + w_j} \biggr   \|  \eta \frac{ \tau }{ \tau + w_j } + (1- \eta) \frac{w_j}{ \tau + w_j} \right) \label{eq:meantrueloss}\\
& \overset{(a)}{\succsim}  n p ( 2 \eta -1)^2  | w_i - \tau |^2
\label{eq:meantrueloss1}
\end{align}
where $(a)$ follows from Pinsker's inequality ($D (p \| q) \geq 2 (p-q)^2$; see~\cite[Theorem 2.33]{yeung2008information} for example) and using the fact that $d_i \asymp np$ when $p > \frac{\log n}{n}$. Here $d_i$ indicates the degree of node $i$: the number of edges incident to node $i$. This suggests that the true point-wise MLE of $w_i$ strictly dominates that of $\tau$ in the \emph{mean} sense. We can actually demonstrate that this is the case beyond the mean sense with high probability, as long as $| w_i - \tau | \succsim \frac{1}{ 2\eta-1} \sqrt{  \frac{ \log n}{  npL} }$ (our hypothesis), which is asserted in the following lemma.

\begin{lemma}
\label{lemma:trueloss}
Suppose that  $| w_i - \tau | \succsim \frac{1}{ 2\eta-1} \sqrt{  \frac{ \log n}{  npL} }$. Then,
\begin{align}
\label{eq:truelossbound}
 \ell^* (w_i) - \ell^* (\tau)    \succsim n p ( 2 \eta -1)^2  | w_i - \tau |^2.
\end{align}
holds with probability approaching one.
\end{lemma}
\begin{proof}
Using Bernstein's inequality formally stated in Lemma~\ref{lemma:BernsteinIneq} (see Appendix~\ref{app:BernsteinIneq}), one can obtain a lower bound on~$\ell^* (w_i) - \ell^* (\tau)$ in terms of  its expectation $\mathbb{E} \left[   \ell^* (w_i) - \ell^* (\tau) | {\cal G} \right]$, its variance $\var \left[   \ell^* (w_i) - \ell^* (\tau) | {\cal G} \right]$, and the maximum value of individual quantities that we sum over. One can then show that the variance and the maximum value are dominated by the expectation under our hypothesis, thus proving that the lower bound is the order of the desired bound as claimed. For completeness, we include the detailed proof at the end of this appendix; see Appendix~\ref{app:trueloss}.
\end{proof}

However, when running our algorithm, we do not have access to the ground truth scores $\boldsymbol{w}_{\backslash i}$. What we can actually compute is 
\begin{equation}
 \hat{\ell} (\tau ):= \frac{1}{L} \log {\cal L} (\tau, \hat{\boldsymbol{w}}_{\backslash i}; \boldsymbol{Y}_{i} )
\end{equation}
 instead of $ {\ell}^* (\tau)$. Fortunately, such surrogate likelihoods are sufficiently close to the true likelihoods, which we will show in the rest of the proof. From this, we will next demonstrate that~\eqref{eq:Ltau_Lwi} holds for sufficiently separated $\tau$ such that $   |\tau -  w_i | \succsim \max
\left \{ \delta + \frac{ \log n}{ n p} \cdot \xi,\frac{1}{ 2\eta -1}   \sqrt{ \frac{ \log n}{ n p L} } \right \}
$.

As seen from~\eqref{eq:trueloss_form2}, one can quantify the difference between $\hat{\ell}(w_i)$ and $\hat{\ell}(\tau)$ as
\begin{align}
\label{eq:surrogateloss}
{\hat{\ell} } (w_i) - {\hat{\ell} } (\tau)    = \sum_{j : (i,j) \in {\cal E}} \left \{
 Y_{ij} \log  \left \{ \frac{ ( \eta w_i + (1 - \eta) {\hat{w}_j } ) ( \eta {\hat{w}_j } + (1 - \eta) \tau ) }{ ( \eta \tau + (1 - \eta) {\hat{w}_j } ) ( \eta {\hat{w}_j } + (1 - \eta) w_i ) }  \right \}
+  \log \left( \frac{ (\tau  + {\hat{w}_j }) ( \eta {\hat{w}_j } + (1-\eta) w_i ) }{ (w_i + {\hat{w}_j }) ( \eta {\hat{w}_j } + (1-\eta) \tau ) } \right)
\right \}.
\end{align}

Using~\eqref{eq:surrogateloss}~and~\eqref{eq:trueloss_form2}, we can represent the gap between the surrogate loss and the true loss as
\begin{align}
& { \hat{\ell }} (w_i) - { \hat{\ell }} (\tau)  - ( \ell^* (w_i) - \ell^* (\tau))   \nn\\
&=\sum_{j : (i,j) \in {\cal E}} Y_{ij} \bigg[\left \{
 \log  \left \{ \frac{ ( \eta w_i + (1 - \eta) {\hat{w}_j } ) ( \eta {\hat{w}_j } + (1 - \eta) \tau ) }{ ( \eta \tau + (1 - \eta) {\hat{w}_j } ) ( \eta {\hat{w}_j } + (1 - \eta) w_i ) }  \right \}   - \log  \left \{ \frac{ ( \eta w_i + (1 - \eta) w_j ) ( \eta w_j + (1 - \eta) \tau ) }{ ( \eta \tau + (1 - \eta) w_j ) ( \eta w_j + (1 - \eta) w_i ) }  \right \}
\right \} \nn\\
&\qquad+
 \left \{
\log \left( \frac{ \tau  + {\hat{w}_j } }{ w_i + {\hat{w}_j } } \right) + \log \left( \frac{  \eta {\hat{w}_j } + (1-\eta) w_i  }{  \eta {\hat{w}_j } + (1-\eta) \tau  } \right) -
\log \left( \frac{ \tau  + w_j }{ w_i + w_j } \right)  -  \log \left( \frac{  \eta {w}_j  + (1-\eta) w_i  }{  \eta {w}_j  + (1-\eta) \tau  } \right)
\right \} \bigg]. 
\label{eqn:gap}
\end{align}
Using Bernstein's inequality under our hypothesis as we did in Lemma~\ref{lemma:trueloss}, one can verify that
\begin{align}
{ \hat{\ell }} (w_i) - { \hat{\ell }} (\tau)  - ( \ell^* (w_i) - \ell^* (\tau)) \lesssim \mathbb{E} \left[ { \hat{\ell }} (w_i) - { \hat{\ell }} (\tau)  - ( \ell^* (w_i) - \ell^* (\tau)) | {\cal G} \right ] = \sum_{j : (i,j) \in {\cal E}} g_\eta( \hat{w}_j)
\end{align}
where
\begin{align}
g_\eta(t): &= \frac{ \eta w_i + (1-\eta) w_j }{w_i + w_j} \bigg\{
 \log  \left \{ \frac{ ( \eta w_i + (1 - \eta) {t } ) ( \eta {t } + (1 - \eta) \tau ) }{ ( \eta \tau + (1 - \eta) {t } ) ( \eta {t } + (1 - \eta) w_i ) }  \right \}   \nn\\*
 &\qquad \qquad - \log  \left \{ \frac{ ( \eta w_i + (1 - \eta) w_j ) ( \eta w_j + (1 - \eta) \tau ) }{ ( \eta \tau + (1 - \eta) w_j ) ( \eta w_j + (1 - \eta) w_i ) }  \bigg \}
\right \} \nn\\
&\qquad + \log \left( \frac{ \tau  + {t } }{ w_i + {t } } \right) + \log \left( \frac{  \eta { t } + (1-\eta) w_i  }{  \eta { t } + (1-\eta) \tau  } \right) -
\log \left( \frac{ \tau  + w_j }{ w_i + w_j } \right)  -  \log \left( \frac{  \eta {w}_j  + (1-\eta) w_i  }{  \eta {w}_j  + (1-\eta) \tau  } \right). 
\label{eqn:def_g}
\end{align}
Here the function $g_\eta(t)$ obeys the following two properties: (i) $g_\eta(w_j)=0$ and (ii)
\begin{align}
\left | \frac{ \partial g_\eta(t) }{ \partial t} \right | &=
\frac{ (2 \eta -1) | \tau-w_i| }{  ( \eta t + (1-\eta) \tau ) ( \eta t + (1-\eta) w_i ) }  \nn\\
&\qquad \times \left | \frac{ \eta w_i + (1-\eta) w_j }{w_i + w_j} \frac{ \eta (1 -\eta) (t^2 - \tau w_i) }{ ( \eta w_i + (1- \eta) t ) ( \eta \tau + (1- \eta) t )  }  - \frac{ \eta t^2 - (1- \eta)w_i \tau }{ (\tau + t) (w_i + t)  } \right | \\
&\overset{(a)}{\lesssim} (2 \eta -1)^2 | \tau - w_i|
\end{align}
where $(a)$ follows from the fact that
\begin{align}
\left | \frac{ \eta w_i + (1-\eta) w_j }{w_i + w_j} \frac{ \eta (1 -\eta) (t^2 - \tau w_i) }{ ( \eta w_i + (1- \eta) t ) ( \eta \tau + (1- \eta) t )  }  - \frac{ \eta t^2 - (1- \eta)w_i \tau }{ (\tau + t) (w_i + t)  } \right | \lesssim (2 \eta -1).
\end{align}
Notice that the left-hand-side in the above is zero when $\eta = 1/2$. This together with the above two properties demonstrates that
\begin{align}
|g_\eta(t) | &\leq | g_\eta(w_j)| + | t-w_j| \cdot \sup_{ t \in [ w_{\min}, w_{\max} ] }  \left | \frac{ \partial g_\eta(t) }{ \partial t} \right | \\
& \lesssim ( 2\eta-1)^2 | \tau - w_i | | t- w_j |.
\end{align}
Applying this to the above gap between the surrogate loss and the true loss, we get:
\begin{align}
\left | { \hat{\ell }} (w_i) - { \hat{\ell }} (\tau)  - ( \ell^* (w_i) - \ell^* (\tau)) \right |  &\lesssim
\sum_{j : (i,j) \in {\cal E}}  ( 2\eta-1)^2 | \tau - w_i | | \hat{w}_j - w_j | 
\label{eq:gaptruesurrogate0}
\\
& \leq  ( 2\eta-1)^2 | \tau - w_i | \sum_{j : (i,j) \in {\cal E}}   | \hat{w}_j^{\rm ub} - w_j |\label{eq:gaptruesurrogate}
\end{align}
where the inequality arises from our hypothesis, namely that  $ |\hat{w}_j - w_j| \leq | \hat{w}_j^{\rm ub} - w_j |$ for all $ j\in [n]$.

We now move on to deriving an upper bound on~\eqref{eq:gaptruesurrogate}. From our assumptions on the initial estimate, we have
\begin{align}
 \| \hat{\boldsymbol{w}} - \boldsymbol{w} \|^2  \leq \| \hat{\boldsymbol{w}}^{\rm ub} - \boldsymbol{w} \|^2 \leq \| \boldsymbol{w} \|^2 \delta^2 \leq n  \delta^2.
\end{align}
Since ${\cal G}$ and $\hat{\boldsymbol{w}}^{\rm ub}$ are statistically independent, this inequality gives rise to:
\begin{align}
& \mathbb{E} \left[  \sum_{j: (i,j) \in {\cal E} } |{\hat{w}_j^{\rm ub}}-w_j| \right] = p \| \hat{\boldsymbol{w}}^{\rm ub} - \boldsymbol{w} \|_1 \leq p\sqrt{n} \| \hat{\boldsymbol{w}}^{\rm ub} - \boldsymbol{w} \| \leq n p \delta, \\
& \mathbb{E} \left[  \sum_{j: (i,j) \in {\cal E} } |{\hat{w}_j^{\rm ub}}-w_j|^2 \right] = p \| \hat{\boldsymbol{w}}^{\rm ub} - \boldsymbol{w} \|^2  \leq np \delta^2.
\end{align}
Recall our assumption that $\max_j | \hat{w}_j^{\rm ub} - w_j | \leq \xi$. Again using Bernstein inequality in Lemma~\ref{lemma:BernsteinIneq} for any fixed $\gamma \geq 3$, with probability at least $1 - 2 n^{-\gamma}$, one has
\begin{align}
\sum_{j: (i,j) \in {\cal E} } |{\hat{w}_j^{\rm ub}}-w_j| &\leq  \mathbb{E} \left[  \sum_{j: (i,j) \in {\cal E} } |{\hat{w}_j^{\rm ub}}-w_j| \right]  + \sqrt{2 \gamma \log n \cdot  \mathbb{E} \left[  \sum_{j: (i,j) \in {\cal E} } |{\hat{w}_j^{\rm ub}}-w_j|^2 \right] } + \frac{2 \gamma}{3} \xi \log n \\
& \leq n p \delta + \sqrt{2 \gamma np \log n    }  \delta+ \frac{2 \gamma}{3} \xi \log n \\
& \overset{(a)}{\leq} n p \delta + \sqrt{ \gamma     } np  \delta+ \frac{2 \gamma}{3} \xi \log n \\
& \overset{(b)}{\leq}  \gamma n p \delta +  \gamma \xi \log n
\end{align}
where $(a)$ follows from our choice on $p$ (we assume $p > \frac{ 2 \log n}{n}$) and $(b)$ follows from the fact that $1 + \sqrt{\gamma} \leq \gamma$ for $\gamma \geq 3$. This combined with~\eqref{eq:gaptruesurrogate} gives us
\begin{align}
\label{eq:truesurrogatelossgap}
\left | { \hat{\ell }} (w_i) - { \hat{\ell }} (\tau)  - ( \ell^* (w_i) - \ell^* (\tau)) \right |
\lesssim  ( 2 \eta -1 )^2  | \tau - w_i | np  \left(  \delta +   \frac{\log n}{np}  \xi \right).
\end{align}

We are now ready to control ${ \hat{\ell }} (w_i) - { \hat{\ell }} (\tau)$. Putting~\eqref{eq:truelossbound} and ~\eqref{eq:truesurrogatelossgap} together, with high probability approaching one, one has
\begin{align}
{ \hat{\ell }} (w_i) - { \hat{\ell }} (\tau) &\succsim \ell^* (w_i) - \ell^* (\tau) - ( 2 \eta -1 )^2  | \tau - w_i | np  \left(  \delta +   \frac{\log n}{np}  \xi \right) \\
&\succsim n p ( 2 \eta -1)^2  | w_i - \tau |^2 -  ( 2 \eta -1 )^2  | \tau - w_i | np  \left(  \delta +   \frac{\log n}{np}  \xi \right) \\
&\succsim 0
\end{align}
where the last step follows from our hypothesis: $  | w_i - \tau |
 \succsim    \delta +   \frac{\log n}{np}  \xi
$. This completes the proof of Lemma~\ref{lemma:l2vslinfty}.

%

\subsection{Proof of Lemma~\ref{lemma:trueloss}}
\label{app:trueloss}

Another representation of the true loss is:
\begin{align}
 \ell^* (w_i) - \ell^* (\tau)  &  = \sum_{j : (i,j) \in {\cal E}} \Bigg\{
 Y_{ij} \log  \left \{ \frac{ ( \eta w_i + (1 - \eta) w_j ) ( \eta w_j + (1 - \eta) \tau ) }{ ( \eta \tau + (1 - \eta) w_j ) ( \eta w_j + (1 - \eta) w_i ) }  \right \}\nn\\*
&\qquad\qquad\qquad+  \log \left( \frac{ (\tau  + w_j) ( \eta w_j + (1-\eta) w_i ) }{ (w_i + w_j) ( \eta w_j + (1-\eta) \tau ) } \right)
\Bigg\}\label{eq:trueloss_form2}
\end{align}
This gives
\begin{align}
 \var \left[ \ell^* (w_i) - \ell^* (\tau)  | {\cal G} \right] &= \var \left[\sum_{j : (i,j) \in {\cal E}}  Y_{ij} \log  \left \{ \frac{ ( \eta w_i + (1 - \eta) w_j ) ( \eta w_j + (1 - \eta) \tau ) }{ ( \eta \tau + (1 - \eta) w_j ) ( \eta w_j + (1 - \eta) w_i ) }  \right \}  \right] \\
 & \overset{(a)}{\lesssim}  | w_i - \tau |^2 (2 \eta -1)^2 \sum_{j : (i,j) \in {\cal E}} \var [Y_{ij} ] \\
 &  =   | w_i - \tau |^2 (2 \eta -1)^2 \sum_{j : (i,j) \in {\cal E}} \frac{1}{L} \frac{ ( \eta w_i + (1-\eta) w_j) ( \eta w_j + (1-\eta) w_i) }{ (w_i + w_j)^2 } \\
 & \lesssim    | w_i - \tau |^2 (2 \eta -1)^2 \frac{np}{L}\label{eq:vartrueloss}
\end{align}
where $(a)$ follows from the fact that $\log \frac{ \beta }{\alpha} \leq \frac{ \beta - \alpha}{\alpha} \textrm{ for } \beta > \alpha > 0$.
Also note that the maximum value of individual quantities $\frac{1}{L}Y_{ij}^{(\ell)}$ that we sum over is given by
\begin{align}
\label{eq:maxtrueloss}
\frac{1}{L } Y_{ij}^{(\ell)} \left | \log  \left \{ \frac{ ( \eta w_i + (1 - \eta) w_j ) ( \eta w_j + (1 - \eta) \tau ) }{ ( \eta \tau + (1 - \eta) w_j ) ( \eta w_j + (1 - \eta) w_i ) }  \right \} \right |    \lesssim  \frac{| w_i - \tau | (2 \eta -1) }{L}.
\end{align}

Making use of Bernstein inequality together with~\eqref{eq:meantrueloss},~\eqref{eq:vartrueloss} and~\eqref{eq:maxtrueloss} suggests that: conditional on ${\cal G}$,
\begin{align}
 \ell^* (w_i) - \ell^* (\tau) &\geq  \mathbb{E} \left[ \ell^* (w_i) - \ell^* (\tau)  | {\cal G} \right]  - \sqrt{2 \gamma \log n \cdot   \var \left[ \ell^* (w_i) - \ell^* (\tau)  | {\cal G} \right] } - \frac{2 \gamma}{3} B \log n  \label{eqn:use_bern1}\\
 & \succsim n p ( 2 \eta -1)^2  | w_i - \tau |^2 - \sqrt{ 2 \gamma} \sqrt{ \frac{ np \log n }{L} } | w_i - \tau | ( 2 \eta -1 ) - \frac{2 \gamma}{3} \frac{| w_i - \tau | (2 \eta -1) }{L}  \log n \label{eqn:use_bern2}\\
 & \geq n p ( 2 \eta -1)^2  | w_i - \tau |^2 - \left( \sqrt{ 2 \gamma} + \frac{2 \gamma}{3} \right) \sqrt{ \frac{ np \log n }{L} } | w_i - \tau | ( 2 \eta -1 ) \label{eqn:use_bern3}\\
 & \overset{(a)}{\succsim} n p ( 2 \eta -1)^2  | w_i - \tau |^2\label{eqn:use_bern4}
\end{align}
holds with probability at least $1 - 2n^{-\gamma}$. Here $(a)$ follows from our hypothesis: $| w_i - \tau | \succsim \frac{1}{ 2\eta-1} \sqrt{  \frac{ \log n}{  npL} }$.

\section{Proof of Lemma~\ref{lemma:l2-norm-bound-etaknown}}
\label{app:ProofofLemmaBoundDelta}

As mentioned earlier, the proof builds  upon the analysis structured by Lemma~2 in~\cite{Negahban2012}, which bounds the deviation of the Markov chain w.r.t.\ the transition matrix $\hat{P}$ (defined in Algorithm~\ref{Algorithm:RC}) after $t$ steps:
\begin{align}
   \frac{  \| \hat{p}_t - { \boldsymbol{w}} \| }{ \| \boldsymbol{w} \| }
\leq \rho^t   \frac{  \| \hat{p}_0 - { \boldsymbol{w}} \| }{ \| \boldsymbol{w} \| }  \sqrt{ \frac{w_{\max}}{w_{\min}} } + \frac{1}{ 1 - \rho} \| \DDelta \|  \sqrt{ \frac{w_{\max}}{w_{\min}} }
\end{align}
where $\hat{p}_t$ denotes the distribution w.r.t. $\hat{P}$ at time $t$ seeded by an arbitrary initial distribution $\hat{p}_0$, the matrix $\DDelta:= \hat{P} - P$ indicates the fluctuation of the transition probability matrix around its mean $P:= \mathbb{E} [\hat{P}]$, and $\rho:=\lambda_{\max} + \| \DDelta \| \sqrt{ \frac{w_{\max}}{w_{\min}} }$.
Here $\lambda_{\max} = \max \{ \lambda_2, - \lambda_n \}$ and $\lambda_i$ indicates the $i$-th eigenvalue of $P$.

   For an arbitrary $\eta$ case, a bound on $\| \DDelta \|$ is: 
\begin{align}
\| \DDelta \| \lesssim \frac{1}{ 2\eta -1} \sqrt{ \frac{ \log n}{ npL } }
\end{align}
which will be proved in the sequel. On the other hand, adapting the analysis in~\cite{Negahban2012} (particularly see Lemma~4 in the reference), one can easily verify that $\rho < 1$ under our assumption that $L n p \succsim \frac{ \log n}{ (2 \eta -1)^2 }$. Applying the bound on $\| \DDelta \|$ and $\rho <1 $ to the above gives the claimed bound, which completes the proof.
 
Let us now prove the bound on $\| \DDelta \|$, which is a generalization of the proof in~\cite{Negahban2012}. Let $D$ be a diagonal matrix with $D_{ii}:= \DDelta_{ii}$. Let $\bar{\DDelta} := \DDelta - D$. Note that
\begin{align}
\| \DDelta \| \leq \| D \| + \| \bar{\DDelta} \| = \max_i | \DDelta_{ii} | + \| \bar{\DDelta} \|.
\end{align}
We will use Hoeffding inequality to bound $|\DDelta_{ii}|$. As for $\| \bar{\DDelta} \|$, we will focus on bounds of  $\mathbb{E} [|\DDelta_{ij}|^p]$, since Tropp inequality in~\cite{Tropp2011} turns out to relate the bound of $\mathbb{E} [|\DDelta_{ij}|^p]$ to that of $\| \bar{\DDelta} \|$, as pointed out in~\cite{Negahban2012}. Hence, here we provide derivations mainly for the bounds on $|\DDelta_{ii}|$ and $\mathbb{E} [|\DDelta_{ij}|^p]$. Later we will appeal to  a relationship between $\| \bar{\DDelta} \|$ and $\mathbb{E} [|\DDelta_{ij}|^p]$, formally stated in Lemma~\ref{lemma:TroppIneq} (see below), to prove the desired bound on $\| \bar{\DDelta} \|$.

\emph{Bounding $|\DDelta_{ii}|$}: Observe that
\begin{align}
L d_{\max } \DDelta_{ii} = -L d_{\max} \sum_{k \neq i} \DDelta_{ik} =  - \sum_{k \neq i}  \sum_{ \ell=1}^{L} \left (
\frac{ Y_{ki}^{(\ell)} - (1 - \eta) }{ 2 \eta -1 } - \frac{w_{k} }{w_i + w_k}  \right)
. \label{eqn:def_Delta_ii}
\end{align}
Let $X_{k \ell} := \frac{ Y_{ki}^{(\ell)} - (1 - \eta) }{ 2 \eta -1 } - \frac{w_{k} }{w_i + w_k}$. Then, we have $\mathbb{E}[X_{k \ell} ] = 0$ and $- \frac{ \eta +1 }{ 2 \eta -1} \leq X_{k \ell} \leq \frac{ \eta }{ 2 \eta -1}$. Using Hoeffding inequality, we obtain:
\begin{align}
 \Pr \left[ | L d_{\max } \DDelta_{ii} | \geq t \right ] \leq 2 \exp \left (  - \frac{2 (t \frac{2 \eta-1}{2 \eta+1})^2 }{  L d_{i} }  \right ) \leq
2 \exp \left (  - \frac{2 (t \frac{2 \eta-1}{2 \eta+1})^2 }{  L d_{\max} }  \right ).\label{eqn:apply_hoeff0}
\end{align}
Choosing  
$t = c \sqrt{ L d_{\max} \log n} \Big(\frac{ 2\eta +1}{ 2 \eta -1}\Big)$ 
 for some $c>0$, one can make the tail bound arbitrarily close to zero in the limit of large $n$. Also $d_{\max } \asymp n p$ when $p > \frac{\log n}{n}$. Hence, with probability approaching one, one has $\|D \| \lesssim  \frac{1}{ 2\eta -1} \sqrt{ \frac{\log n}{ np L } }$.


\emph{Bounding $\| \bar{ \DDelta} \|$}: A careful inspection reveals that
\begin{align}
\bar{ \DDelta} = \sum_{i < j: (i,j) \in {\cal E} } (e_i e_j^T - e_j e_i^T ) (\hat{p}_{ij} - p_{ij})
\end{align}
where $e_i$ denotes the standard basis vector in which only the $i$-th entry is 1 while the others are zeros. Here with a slight abuse of notation, we use ${\cal E}$ to indicate ${\cal E}^{\sf init}$.
As mentioned earlier, we intend to make use of the concentration result by Tropp~\cite{Tropp2011} for sum of independent self-adjoint matrices. To this end, we apply the dilation idea in~\cite{Tropp2011} for symmetrization:
\begin{align}
Z_{ij} := A_{ij} \DDelta_{ij} :=
\left[
  \begin{array}{cc}
    0          & e_i e_j^T - e_j e_i^T  \\
   e_j e_i^T - e_i e_j^T & 0  \\
  \end{array}
\right] \DDelta_{ij}.
\end{align}
Note that
\begin{align}
\| \bar{ \DDelta} \| = \left \| \sum_{i < j: (i,j) \in {\cal E} } (e_i e_j^T - e_j e_i^T ) (\hat{p}_{ij} - p_{ij}) \right \| =
\left \| \sum_{i < j: (i,j) \in {\cal E} } A_{ij} \DDelta_{ij}  \right \| = \left \| \sum_{i < j : (i,j) \in {\cal E}} {Z}_{ij}  \right \|.
\end{align}
We now invoke Tropp's inequality formally stated in the following lemma.

\begin{lemma}
\label{lemma:TroppIneq}
Consider a sequence $Z_{ij}$ of independent random self-adjoint matrices. Assume that
\begin{align}
\mathbb{E} [ {Z}_{ij}] = 0 \quad \textrm{and} \quad \mathbb{E} [ {Z}_{ij}^p ] \preceq \frac{p !}{2} R^{p-2} \tilde{A}_{ij}^2, \quad p \geq 2.
\end{align}
Define $\sigma^2:= \left \| \sum_{i,j} \tilde{A}_{ij}^2 \right \|
$. Then, for all $t \geq 0$,
\begin{align}
\Pr
\bigg [   \Big \| \sum_{i,j} {Z}_{ij} \Big \| \geq t\bigg]
\leq \exp \left (  - \frac{ t^2/2}{ \sigma^2 + Rt }  \right ).
\end{align}
\end{lemma}
To figure out what $\tilde{A}_{ij}$, $\sigma^2$ and $R$ are, we consider
\begin{align}
 \mathbb{E}[Z_{ij}^p ] & \overset{(a)}{\preceq} \mathbb{E} [  \Delta_{ij}^p ]   A_{ij}^2  \\
& \overset{(b)}{\leq}
\frac{p! }{2} \left ( \frac{ 2 \eta +1 }{ 2 \eta -1} \frac{1}{ \ \sqrt{ L d_{\max}^2 } } \right)^{p-2} \frac{ 2 \eta -1}{ 2\eta +1 } \frac{1 }{  L d_{\max}^2    }A_{ij}^2.
\end{align}
To see $(a)$, note that $A_{ij}^p$ is equal to $A_{ij}^2$ when $p$ is even; $A_{ij}$ otherwise. Also one can verify that the eigenvalues of $A_{ij}$ are either 1 or $-1$. Hence, $A_{ij}^p \preceq A_{ij}^2$. To see $(b)$, observe that 
\begin{align}
 L d_{\max} \DDelta_{ij}  =  \sum_{ \ell=1}^{L} \left ( \frac{ Y_{ji}^{(\ell)} - (1 - \eta) }{ 2 \eta -1 } - \frac{w_{j} }{w_i + w_j} \label{eqn:def_Delta_ij}
\right).
\end{align}
Applying Hoeffding inequality into the term inside the summation, we get 
\begin{equation}
\Pr \left[ | L d_{\max } \DDelta_{ij} | \geq t \right ] \leq 2 \exp \left (  - \frac{2 (t \frac{2 \eta-1}{2 \eta+1})^2 }{  L }  \right ), \label{eqn:apply_hoeff1}
\end{equation}
 which yields 
\begin{equation}
 \Pr \left[ | \DDelta_{ij} | \geq t \right ] \leq 2 \exp \left (  - 2 t^2 \left(  \frac{2 \eta-1}{2 \eta+1}\right)^2 L d_{\max}^2   \right ).\label{eqn:apply_hoeff2}
 \end{equation}
This implies that $\Delta_{ij}$ is a sub-Gaussian random variable. Hence, wet get:
\begin{align}
\label{eq:BoundDeltaijp}
 \mathbb{E} \left [ | \DDelta_{ij} |^p \right] \leq \frac{p!}{2} \left( \frac{ 2 \eta + 1}{ 2 \eta -1 } \frac{1}{ \sqrt{ L d_{\max}^2 } } \right)^{p},
\end{align}
which yields $(b)$.

%
We now see that
 $R = \frac{ 2 \eta +1 }{ 2 \eta -1} \frac{1}{ \ \sqrt{ L d_{\max}^2 } }$ and $\tilde{A}_{ij}^2 = \frac{ 2 \eta -1}{ 2\eta +1 } \frac{1 }{  L d_{\max}^2    }A_{ij}^2$.
Some calculation gives
\begin{align}
\sigma_2 &:= \left \| \sum_{i < j : (i,j) \in \calE} \tilde{A}_{ij}^2  \right \| \\
&  =  \frac{ 2 \eta -1}{ 2\eta +1 } \frac{1 }{  L d_{\max}^2    } \left \| \sum_{i=1}^n \sum_{j=i+1}^{n} {\bf 1}
 \left \{  (i,j) \in \calE\right \} \left[
  \begin{array}{cc}
    e_i e_i^T + e_j e_j^T      &  0  \\
    0  & e_i e_i^T + e_j e_j^T  \\
  \end{array}
\right]
  \right \| \\
  & = \frac{ 2 \eta -1}{ 2\eta +1 } \frac{1 }{  L d_{\max}^2    }
\left \|
\sum_{i=1}^n d_i
\left[
  \begin{array}{cc}
    e_i e_i^T       &  0  \\
    0  & e_i e_i^T  \\
  \end{array}
\right]
\right \| \\
& = \frac{ 2 \eta -1}{ 2\eta +1 } \frac{1 }{  L d_{\max} }.
\end{align}
Now applying Lemma~\ref{lemma:TroppIneq} and using the fact that $\| \bar{ \DDelta} \| = \| \sum_{i,j} Z_{ij} \|$, we get:
\begin{align}
\Pr
\left [   \left \| \bar{\DDelta} \right \| \geq t\right]
\leq 2n \exp \left (  - \frac{ t^2/2}{ \frac{ (2 \eta +1)^2 }{  L d_{\max} (2\eta -1 )^2   }
 + \left ( \frac{2 \eta +1 }{ \sqrt{ L d_{\max}^2 (2 \eta -1 )^2 } } \right)t }  \right ).
\end{align}
Under the assumption that $ d_{\max} \asymp np \geq \log n$ and choosing $t = \frac{c_1}{ 2 \eta-1} \sqrt{ \log n/ (npL)}$, the tail probability is bounded by $ 2 n \exp \{  - c_2^2 \log n \}$ for some constants $c_1$ and $c_2$. Hence, with probability approaching one, we get the desired bound:
\begin{align}
 \| \bar{\DDelta} \| \lesssim \frac{1}{ 2 \eta -1 } \sqrt{ \frac{ \log n }{ npL}}.
\end{align}


\section{Proof of Lemma \ref{lemma:l2vslinfty_est}}\label{app:lemma:l2vslinfty_est}
\begin{proof} 
In this proof, for the sake of brevity, we only highlight the parts of the proof of Lemma \ref{lemma:l2vslinfty} that have to be modified when we use $\hat{\eta}$ in place of $\eta$ in the likelihood function.

  Define
\begin{align}
{\hel}^*(\tau)& := \frac{1}{L}\log\hat{\calL}(\tau,\bw_{\setminus i}; \bY_i^{\mathrm{iter}})  \\
&=\sum_{j : (i,j) \in \calE} \left\{  Y_{ij} \log\left( \hat{\eta}\frac{\tau}{\tau + w_j} + (1-\hat{\eta})\frac{w_j}{\tau + w_j}\right) + (1-Y_{ij}) \log \left( \hat{\eta} \frac{w_j}{\tau+w_j} + (1-\hat{\eta}) \frac{\tau}{\tau+w_j}\right)\right\} \label{eqn:def_Lhat} .
\end{align}
Notice that $\hel^*$ is similar to $\ell^*$ in \eqref{eqn:likelihood_fn} except that $\eta$ in the latter is replaced by its surrogate $\hat{\eta}$  in the former because we only have access to this estimate.

Consider the difference
\begin{align}
{\hel}^*(w_i) -{\hel}^*(\tau)  &= \sum_{j :(i,j) \in\calE} \Bigg\{ Y_{ij} \log\left( \frac{ \hat{\eta}\frac{w_i}{w_i+ w_j} + (1-\hat{\eta})\frac{w_j}{w_i+ w_j}}{ \hat{\eta}\frac{\tau}{\tau + w_j} + (1-\hat{\eta})\frac{w_j}{\tau + w_j}}\right) \nn\\*
&\qquad\qquad\qquad+ (1-Y_{ij}) \log\left(  \frac{ \hat{\eta}\frac{w_j}{w_i+ w_j} + (1-\hat{\eta})\frac{w_i}{w_i+ w_j}}{ \hat{\eta}\frac{w_j}{\tau + w_j}+ (1-\hat{\eta}) \frac{\tau}{\tau+w_j}}   \right)\Bigg\} \label{eqn:diff_l} .
\end{align}
Now when we take expectation
\begin{equation}
\bbE[ Y_{ij}] = \eta\frac{w_i}{w_i+w_j} + (1-\eta) \frac{w_j}{w_i+w_j} .\label{eqn:expect_Y}
\end{equation}
Note that this is in terms of $\eta$ and not $\hat{\eta}$ as in the difference of the empirical log-likelihoods in \eqref{eqn:diff_l}.  In particular, $\bbE [ {\hel}^*(w_i) -{\hel}^*(\tau) \,|\, \calG]$ is not a sum of KL divergences but instead there is some ``mismatch''. However, by   some basic approximations,    we have
\begin{align}
&\bbE  \left[ {\hel}^*(w_i) -{\hel}^*(\tau) \, \big|\, \calG \right] \nn\\*
&=\sum_{j : (i,j)\in\calE} \Bigg\{ \left( \eta\frac{w_i}{w_i+w_j} + (1-\eta) \frac{w_j}{w_i+w_j} \right) \log\left( \frac{ \hat{\eta}\frac{w_i}{w_i+ w_j} + (1-\hat{\eta})\frac{w_j}{w_i+ w_j}}{ \hat{\eta}\frac{\tau}{\tau + w_j} + (1-\hat{\eta})\frac{w_j}{\tau + w_j}}\right) \nn\\*
&\qquad\qquad\qquad+ \left( \eta\frac{w_j}{w_i+w_j} + (1-\eta) \frac{w_i}{w_i+w_j} \right) \log\left(  \frac{ \hat{\eta}\frac{w_j}{w_i+ w_j} + (1-\hat{\eta})\frac{w_i}{w_i+ w_j}}{ \hat{\eta}\frac{w_j}{\tau + w_j}+ (1-\hat{\eta}) \frac{\tau}{\tau+w_j}} \right) \Bigg\}\label{eqn:plug_exp}  \\ 
&\succsim  \sum_{j : (i,j)\in\calE}  \Bigg\{\left( \heta\frac{w_i}{w_i+w_j} + (1-\heta) \frac{w_j}{w_i+w_j} \right) \log\left( \frac{ \hat{\eta}\frac{w_i}{w_i+ w_j} + (1-\hat{\eta})\frac{w_j}{w_i+ w_j}}{ \hat{\eta}\frac{\tau}{\tau + w_j} + (1-\hat{\eta})\frac{w_j}{\tau + w_j}}\right) \nn\\*
&\qquad\qquad\qquad+ \left( \heta\frac{w_j}{w_i+w_j} + (1-\heta) \frac{w_i}{w_i+w_j} \right) \log\left(  \frac{ \hat{\eta}\frac{w_j}{w_i+ w_j} + (1-\hat{\eta})\frac{w_i}{w_i+ w_j}}{ \hat{\eta}\frac{w_j}{\tau + w_j}+ (1-\hat{\eta}) \frac{\tau}{\tau+w_j}} \right)  \Bigg\} \label{eqn:approx_mult}\\  
&=\sum_{j : (i,j)\in\calE}  D\left( \heta\frac{w_i}{w_i+w_j}+(1-\heta)\frac{w_j}{w_i+w_j} \,\Big\|\, \heta\frac{\tau}{\tau+w_j}+(1-\heta)\frac{w_j}{\tau+w_j} \right) \label{eqn:div} \\
&\succsim np(2\heta-1)^2 |w_i-\tau|^2  \label{eqn:pin2} 
\end{align}
where
\begin{enumerate}
\item \eqref{eqn:plug_exp} follows from the difference of $\kappa^*$'s in \eqref{eqn:diff_l} and the expectation in \eqref{eqn:expect_Y};
\item \eqref{eqn:approx_mult} holds with high probability (guaranteed by the sample complexity bound in Theorem \ref{thm:etaunknown}) by multiplicatively and uniformly approximating $\eta w_i + (1-\eta)w_j$ by $\heta w_i + (1-\heta)w_j$ and $\eta w_j + (1-\eta)w_i$ by $\heta w_j + (1-\heta)w_i$ using Lemma \ref{lem:est_ratio}  (in Appendix \ref{sec:auxlemma1} at the end of this appendix)   with constant $\nu=0.1$ (say);
\item \eqref{eqn:pin2} is an application of Pinsker's inequality~\cite[Theorem 2.33]{yeung2008information}.
\end{enumerate}
%
The punchline in this calculation is that with our choice of parameters, the scaling of the lower bound of $\bbE [ {\hel}^*(w_i) -{\hel}^*(\tau) \,|\, \calG]$ is the same as that for the known $\eta$ case in \eqref{eq:meantrueloss1}.

Now we bound  the conditional variance.  We have
\begin{align}
&\var  \left[ {\hel}^*(w_i) -{\hel}^*(\tau) \, \big|\, \calG \right] \nn\\*
 &= \var\left[ \sum_{ j: (i,j) \in\calE} Y_{ij} \log \left\{  \frac{ (\hat{\eta} w_i +(1-\hat{\eta})w_j)( \hat{\eta} w_j + (1-\hat{\eta}) \tau)}{ (\hat{\eta} \tau + (1-\hat{\eta} )w_j )(\hat{\eta } w_j  + (1-\hat{\eta}  )w_i )} \right\}\right] \\
&\lesssim |w_i - \tau|^2 (2\hat{\eta}-1)^2 \sum_{ j : (i,j) \in\calE}  \var[Y_{ij}] \label{eqn:follows_ori}\\
&\le |w_i - \tau|^2 (2\hat{\eta}-1)^2 \sum_{ j : (i,j) \in\calE}   \frac{1}{4L} \label{eqn:bd_ber}\\
&\lesssim |w_i - \tau|^2 (2\hat{\eta}-1)^2 \frac{np}{L}. \label{eqn:high_pr_gr}
\end{align}
where
\begin{enumerate}\item \eqref{eqn:follows_ori} follows from the original argument as in the proof of Lemma \ref{lemma:trueloss} in Appendix \ref{app:trueloss};\item \eqref{eqn:bd_ber} follows from the fact that the variance of any Bernoulli random variable is upper bounded by $1/4$;\item and \eqref{eqn:high_pr_gr} holds with high probability due to the nature of the Erd\H{o}s-R\'enyi graph.\end{enumerate}
Thus, by using the bounds in \eqref{eqn:pin2},  \eqref{eqn:high_pr_gr} and Bernstein's inequality (Lemma \ref{lemma:BernsteinIneq}), and mimicking the proof of Lemma \ref{lemma:trueloss} in Appendix~\ref{app:trueloss}   with $\hat{\eta}$ in place of $\eta$, we may conclude that
\begin{equation}
\hel^*(w_i) -\hel^*(\tau)\succsim np(2\hat{\eta}-1)^2 |w_i-\tau|^2 .
\end{equation}
By Lemma \ref{lem:est_eta1} which allows us to  multiplicatively approximate $(2\hat{\eta}-1)^2 $ with $(2 {\eta}-1)^2 $ (to within a constant factor of $(1-\nu)^2$), we also have
\begin{equation}
\hel^*(w_i) -\hel^*(\tau)\succsim np(2 {\eta}-1)^2 |w_i-\tau|^2
\end{equation}
with probability tending to one polynomially fast.

 Just as in the proof of  Lemma \ref{lemma:l2vslinfty}, we do not have access to the true ground truth scores  $\bw_{\setminus i}$. We instead analyze the behavior of surrogate log-likelihoods $\hhel$ with  the true score vectors $\bw_{\setminus i}$ replaced by their estimates $\hat{\bw}_{\setminus i}$. We have
\begin{align}
\hhel(w_i )-\hhel(\tau )&=\sum_{j : (i,j) \in {\cal E}} \Bigg\{
 Y_{ij} \log  \left \{ \frac{ ( \heta w_i + (1 - \heta) {\hat{w}_j } ) ( \heta {\hat{w}_j } + (1 - \heta) \tau ) }{ ( \heta \tau + (1 - \heta) {\hat{w}_j } ) ( \heta {\hat{w}_j } + (1 - \heta) w_i ) }  \right \}
\nn\\*
&\qquad\qquad\qquad+  \log \bigg\{ \frac{ (\tau  + {\hat{w}_j }) ( \heta {\hat{w}_j } + (1-\heta) w_i ) }{ (w_i + {\hat{w}_j }) ( \heta {\hat{w}_j } + (1-\heta) \tau ) } \bigg\}
 \Bigg\}.
\end{align}
In a similar way to the case where $\eta$ is known (cf.\ \eqref{eqn:gap}), we can quantify the gap between the difference of surrogate log-likelihoods $ \hhel(w_i )-\hhel(\tau ) $ and difference of true log-likelihoods $\hel^*  (w_i )-\hel^*(\tau )$  as follows:
\begin{equation} \label{eqn:four_h}
\hhel(w_i )-\hhel(\tau )-\big( \hel^*  (w_i )-\hel^*(\tau ) \big) \lesssim\sum_{j: (i,j)\in\calE} g_{\eta,\heta}(\hatw_j),
\end{equation}
where now
\begin{align}
g_{\eta,\heta}(t)&:= \frac{ \eta w_i + (1-\eta) w_j }{w_i + w_j} \Bigg \{
 \log  \left ( \frac{ ( \heta w_i + (1 - \heta) {t } ) ( \heta {t } + (1 - \heta) \tau ) }{ ( \heta \tau + (1 - \heta) {t } ) ( \heta {t } + (1 - \heta) w_i ) }  \right ) \nn\\*
&\qquad   \qquad - \log  \left ( \frac{ ( \heta w_i + (1 - \heta) w_j ) ( \heta w_j + (1 - \heta) \tau ) }{ ( \heta \tau + (1 - \heta) w_j ) ( \heta w_j + (1 - \heta) w_i ) }  \right)
\Bigg\} \nn\\
&\qquad + \log \left( \frac{ \tau  + {t } }{ w_i + {t } } \right) + \log \left( \frac{  \heta { t } + (1-\heta) w_i  }{  \heta { t } + (1-\heta) \tau  } \right) -
\log \left( \frac{ \tau  + w_j }{ w_i + w_j } \right)  -  \log \left( \frac{  \heta {w}_j  + (1-\heta) w_i  }{  \heta {w}_j  + (1-\heta) \tau  } \right). \label{eqn:getaheta}
\end{align}
Note that $g_{\eta,\eta}(t)=g_\eta (t)$    in~\eqref{eqn:def_g}  in the proof of Lemma \ref{lemma:l2vslinfty}. The reason why $\eta$ appears in the leading factor in  \eqref{eqn:getaheta} is because we are taking expectation of $Y_{ij}$   which is generated from the {\em true} model with parameter $\eta$ (cf.\ \eqref{eqn:expect_Y}). The parameter $\heta$ appears in $\{\ldots\}$ in \eqref{eqn:getaheta} because the log-likelihood function ${\hel}^*(\cdot)$ (cf.\ \eqref{eqn:def_Lhat}) is defined with respect to the surrogate $\heta$  since here we assume we have no knowledge of the true $\eta$.

Several properties of $g_\eta(t)$ were studied in the proof of Lemma \ref{lemma:l2vslinfty}. Here we need to study $g_{\eta,\heta}(t)$.  In fact, by using Lemma \ref{lem:est_ratio} to approximate $\eta w_i + (1-\eta) w_j$ with $\heta w_i + (1-\heta) w_j$, we see that with probability tending to one polynomially fast,
\begin{align}
g_{\eta,\heta}(t) &\lesssim\frac{ \heta w_i + (1-\heta) w_j }{w_i + w_j} \bigg \{
 \log  \left \{ \frac{ ( \heta w_i + (1 - \heta) {t } ) ( \heta {t } + (1 - \heta) \tau ) }{ ( \heta \tau + (1 - \heta) {t } ) ( \heta {t } + (1 - \heta) w_i ) }  \right \} \nn\\*
&\qquad   \qquad - \log  \left \{ \frac{ ( \heta w_i + (1 - \heta) w_j ) ( \heta w_j + (1 - \heta) \tau ) }{ ( \heta \tau + (1 - \heta) w_j ) ( \heta w_j + (1 - \heta) w_i ) }  \bigg\}
\right \} \nn\\
&\qquad + \log \left( \frac{ \tau  + {t } }{ w_i + {t } } \right) + \log \left( \frac{  \heta { t } + (1-\heta) w_i  }{  \heta { t } + (1-\heta) \tau  } \right) -
\log \left( \frac{ \tau  + w_j }{ w_i + w_j } \right)  -  \log \left( \frac{  \heta {w}_j  + (1-\heta) w_i  }{  \heta {w}_j  + (1-\heta) \tau  } \right)  \label{eqn:getaheta_up} \\
&= g_{\heta}(t)
\end{align}
where $g_{\heta}(t)$ is $g(t)$ in~\eqref{eqn:def_g} with $\eta$ replaced by $\heta$.  Basically, we replaced the factor $\heta w_i + (1-\heta) w_j $ with (a constant multiplied by) $\eta w_i + (1-\eta) w_j$ in \eqref{eqn:getaheta_up}. Now, the bound in \eqref{eqn:four_h} can be further upper bounded as
\begin{equation}
\hhel(w_i )-\hhel(\tau )-\big( \hel^*  (w_i )-\hel^*(\tau ) \big) \lesssim\sum_{j: (i,j)\in\calE} g_{\heta}(\hatw_j).
\end{equation}
The rest of the proof of Lemma \ref{lemma:l2vslinfty}, in particular the steps in~\eqref{eqn:use_bern1}--\eqref{eqn:use_bern4}, goes  through verbatim with $\eta$ replaced by $\heta$. Finally,  we can use Lemma  \ref{lem:est_eta1} to multiplicatively approximate $(2\heta-1)$ with $(2\eta-1)$  to complete the proof of Lemma \ref{lemma:l2vslinfty_est}. \end{proof}

\subsection{Approximation Lemmata and Their Proofs}\label{sec:auxlemma1}

\begin{lemma} \label{lem:est_ratio}
For any pair of weights $( w_i, w_j)$ and any
constant $\nu>0$, if
\begin{equation}
L\succsim \left(\frac{w_{\max}}{\nu w_{\min}}\right)^2\log\frac{n}{\delta} , \label{eqn:Lmult_lb}
\end{equation}
we have that
\begin{equation}
 \bigg|\Big(\frac{   \eta w_i + (1-\eta) w_j }{ \hat{\eta}w_i + (1-\hat{\eta}) w_j  } \Big)-1 \bigg|\le \nu \label{eqn:bd_ratio}
\end{equation}
with probability exceeding $1-\delta$.
\end{lemma}
The important point here is that this approximation is {\em uniform} over $(i,j)\in [n]^2$; cf.\ the lower bound on $L$ in \eqref{eqn:Lmult_lb} and the threshold $\nu$ in~\eqref{eqn:bd_ratio}  does not depend on $(i,j)$. This bound implies that, with high probability, we can readily approximate $\eta w_i + (1-\eta) w_j$ with $(1\pm \nu) (\hat{\eta}w_i + (1-\hat{\eta}) w_j)$ for any constant $\nu>0$.  Also note that since $w_{\min},w_{\max}=\Theta(1)$ and $\nu>0$ is also a constant, the bound in \eqref{eqn:Lmult_lb} is in fact $L\succsim \log\frac{n}{\delta}\asymp\log n$ (with $\delta= 1/\poly(n)$). This is clearly satisfied by the assumption  in  \eqref{eq:MinSampleComplexity_unknown} in   Theorem \ref{thm:etaunknown}.

\begin{proof}[Proof of Lemma \ref{lem:est_ratio}]
Assume without loss of generality that $w_i>w_j$ (the expression in  \eqref{eqn:bd_ratio} is symmetric in $w_i$ and $w_j$). Consider
\begin{align}
\Pr \left(  \frac{   \eta w_i + (1-\eta) w_j }{ \hat{\eta}w_i + (1-\hat{\eta}) w_j  } > 1+\nu \right) &= \Pr\big ( \eta w_i + (1-\eta) w_j > (1+\nu)  (\hat{\eta} w_i + (1-\hat{\eta}) w_j  ) \big) \\
&= \Pr \big( (\eta-\hat{\eta})(w_i - w_j) > \nu\hat{\eta}w_i +  \nu(1-\hat{\eta})w_j \big)\\
&\le \Pr \big( (\eta-\hat{\eta})(w_i - w_j) > \nu w_{\min} \big) \label{eqn:lower_bd_w}\\
&= \Pr \bigg(  \eta-\hat{\eta} > \nu \frac{w_{\min}}{w_i - w_j } \bigg) \label{eqn:assume_wiwj}\\
&\le \Pr \bigg(  \eta-\hat{\eta} > \nu \frac{w_{\min}}{w_{\max}} \bigg) \label{eqn:bound_wmax}\\
&\le \Pr \bigg( |\eta-\hat{\eta} |> \nu \frac{w_{\min}}{w_{\max}} \bigg) \label{eqn:bound_abs}
\end{align}
where in~\eqref{eqn:lower_bd_w}, we lower bounded $w_i,w_j$ by $w_{\min}$,  \eqref{eqn:assume_wiwj} assumes that $w_i >w_j$ and \eqref{eqn:bound_wmax} follows because $w_i-w_j\le w_i\le w_{\max}$. A bound  for the other inequality $\Pr \big(  \frac{   \eta w_i + (1-\eta) w_j }{ \hat{\eta}w_i + (1-\hat{\eta}) w_j  } < 1-\nu \big)$ proceeds in a completely analogous   way. Since $w_{\min},w_{\max}=\Theta(1)$, the result follows immediately from the union bound and the probabilistic bound on $|\hat{\eta}-\eta|$ (Lemma \ref{lem:fidelity}).
\end{proof}
\begin{lemma} \label{lem:est_eta1}
For any constant $\nu>0$, if
\begin{equation}
L\succsim \frac{1}{\nu^2(2\eta-1)^2}\log\frac{n}{\delta}, \label{eqn:L_lb} 
\end{equation}
we have that
\begin{equation}
\bigg|\Big(\frac{2\hat{\eta}-1}{2\eta-1} \Big)-1 \bigg|\le \nu \label{eqn:ratio_eta}
\end{equation}
with probability exceeding $1-\delta$.
\end{lemma}
Here, in contrast to Lemma \ref{lem:est_ratio},  $(2\eta-1)$  in \eqref{eqn:ratio_eta} may be vanishingly small, so the lower bound on $L$ in \eqref{eqn:L_lb} contains the additional term $(2\eta-1)^2$. 
\begin{proof}[Proof of Lemma \ref{lem:est_eta1}]
Consider
\begin{align}
\Pr\left( \bigg|\Big(\frac{2\hat{\eta}-1}{2\eta-1} \Big)-1 \bigg| >\nu \right) &= \Pr\left( \bigg| \frac{\hat{\eta}-\eta}{2\eta-1}\bigg|>\frac{\nu}{2}\right)\\
&= \Pr\left(  \big| \hat{\eta}-\eta \big|>\frac{\nu}{2}(2\eta-1)\right) \label{eqn:mult_err}.
\end{align}
But we know from  Lemma \ref{lem:fidelity} that if
\begin{equation}
L\succsim \frac{1}{\big(\frac{\nu}{2}(2\eta-1)\big)^2}\log\frac{n}{\delta}\asymp\frac{1}{\nu^2 (2\eta-1)^2 } \log\frac{n}{\delta},
\end{equation}
then the probability in \eqref{eqn:mult_err} is no larger than $\delta$.
\end{proof}


\section{Proof of Lemma \ref{lem:l2_error_unknown}}\label{app:prf_upsilon_unknown}
From the proof sketch in Section \ref{sec:prf_lem:l2_error_unknown}, we see that it suffices to prove the upper bound on $\|\hat{\DDelta}\|$ in \eqref{eqn:bd_upsilon}.   The entries of $\hat{\DDelta}$ are denoted in the usual way as $\hat{\DDelta}_{ij}$ where $i,j\in [n]$.
When $\eta$ was known, it was imperative to understand the probability that
\begin{equation}
F_{ij}:=Ld_{\max}\DDelta_{ij}=\frac{\big(\sum_{\ell=1}^L Y_{ij}^{(\ell)} \big)- L(1-\eta) }{2\eta-1}- L\frac{w_i}{w_i+w_j} \label{eqn:Cactual}
\end{equation}
deviates from zero. See the corresponding bound in \eqref{eqn:apply_hoeff1}.   
When one only has an estimate of $\eta$, namely $\hat{\eta}$, it is then imperative to do the same for \begin{equation}
\hatF_{ij}:= \frac{\big(\sum_{\ell=1}^L Y_{ij}^{(\ell)} \big)- L(1- \hat{\eta}) }{2\hat{\eta}-1}- L\frac{w_i}{w_i+w_j}. \label{eqn:Cest}
\end{equation}
Our overarching strategy is to bound $\hatF_{ij}$ in terms of $F_{ij}$ and then use the concentration bound  we had established for $F_{ij}$ in \eqref{eqn:apply_hoeff1} to  then understand the stochastic behavior of $\hatF_{ij}$. To simplify notation, define the sum $U:=LY_{ij} = \sum_{\ell=1}^L Y_{ij}^{(\ell)} $. Consequently,
\begin{align}
\big| \hatF_{ij} - F_{ij} \big| &= \left| \frac{U- L(1-\hat{\eta}) }{2\hat{\eta}-1}-\frac{U-L(1-\eta)}{2\eta-1} \right|\\
&\le L \left| \frac{1-\heta}{2\heta-1}-\frac{1-\eta}{2\eta-1}\right| + U\left| \frac{1 }{2\hat{\eta}-1}-\frac{1 }{2\eta-1} \right|\\
&\le L \bigg[\,  \left| \frac{1-\heta }{2\hat{\eta}-1}-\frac{1-\eta }{2\eta-1}  \right| + \left| \frac{1 }{2\hat{\eta}-1}-\frac{1 }{2\eta-1} \right| \, \bigg]
\end{align}
where the final bound follows from the fact that $|U |\le L$ almost surely (since $Y_{ij}^{(\ell)}\in\{0,1\}$).  Now we make use of the following   lemma that uses the sample complexity result in  Lemma \ref{lem:fidelity} to  quantify the Lipschitz constant of the maps $t\mapsto\frac{1}{ 2t-1 }$ and $t\mapsto\frac{1-t}{2t-1}$ in the vicinity of $t=(1/2)^+$.

\begin{lemma}\label{lem:est_eta}
Let $\lambda_1: (1/2,1]\to\bbR_+$ and $\lambda_2 : (1/2,1]\to\bbR_+$ be defined as
\begin{equation}
\lambda_1(t) := \frac{1-t}{2t-1},\quad \mbox{and}\quad \lambda_2(t) := \frac{1}{2t-1} \label{eqn:gt} . 
\end{equation}
Then if
\begin{equation}
L\succsim \frac{1}{(2\eta-1)^2}\log\frac{n}{\delta}  \label{eqn:Lcond0}
\end{equation}
with probability exceeding $1-\delta$ (over the random variable $\hat{\eta}$ which depends on the samples drawn from the mixture distribution  \eqref{eqn:mixture_mode}), we have for each $j = 1,2$,
\begin{equation}
|\lambda_j(\hat{\eta}) - \lambda_j(\eta)| \leq\frac{8}{(2\eta-1)^2}|\hat{\eta}-\eta| .
\end{equation}
\end{lemma}
The proof of this lemma is deferred to Appendix \ref{app:prf_est_eta} at the end of this appendix. We take $\delta = 1/\poly(n)$ in the sequel so \eqref{eqn:Lcond0} is equivalently
\begin{equation}
L\succsim \frac{\log n}{(2\eta-1)^2} \label{eqn:Lcond}
\end{equation}
which  when combined with $S =  \binom{n}{2}pL$ is less stringent than the statement of Theorem \ref{thm:etaunknown}.
Thus,  under the condition \eqref{eqn:Lcond},  Lemma \ref{lem:est_eta} yields that
\begin{equation}
\big| \hatF_{ij} - F_{ij} \big|  \le   \frac{16L}{(2\eta-1)^2} |\hat{\eta}-\eta|   \label{eqn:diff_F}
\end{equation}
with probability exceeding $1-1/\poly(n)$.
By the reverse triangle inequality, we obtain
\begin{equation}
\big| \hatF_{ij} - F_{ij} \big|\ge \big|  |\hatF_{ij}|- |F_{ij}|\big|.\label{eqn:rev_F}
\end{equation}
To make the dependence of $|\hat{\eta}-\eta|$ on the number of samples $L$ explicit, we define
\begin{equation}
\eps_L:= |\hat{\eta} -\eta|. \label{eqn:def_epsL0}
\end{equation}
By uniting \eqref{eqn:diff_F}--\eqref{eqn:def_epsL0}, we obtain
\begin{equation}
 |F_{ij} | - \eps_L' \le \big|\hatF_{ij}\big| \le |F_{ij}|+\eps_L' \label{eqn:FF}
\end{equation}
where
\begin{equation}
\eps_L':=\frac{16L}{(2\eta-1)^2} \eps_L. \label{eqn:def_epsLp}
\end{equation}
For later reference, define
\begin{equation}
\eps_L'':=\frac{16L}{(2\eta-1)^2}d_{\max} \eps_L. \label{eqn:def_epsLpp}
\end{equation}
With the estimate in \eqref{eqn:FF}, we observe that for any $t>0$, one has
\begin{equation}
\Pr\left[ \big|\hatF_{ij}\big|\ge t \right]\le \Pr\left[  |F_{ij} |+\eps_L'\ge t \right]= \Pr\left[  |F_{ij} |\ge t-\eps_L' \right] \label{eqn:bd_F}
\end{equation}
where the randomness in the probability on the left is over both $\hat{\eta}$ and $\bY:=\{Y_{ij}^{(\ell) } : \ell \in [L], (i,j)\in\calE\}$ (the former is a function of the latter) whereas the randomness in the probability on the  right is only over  $\bY$. Thus, by using the equality $F_{ij}=Ld_{\max} {\DDelta}_{ij} $ and applying Hoeffding's inequality  to \eqref{eqn:bd_F} (cf.\ the bound in~\eqref{eqn:apply_hoeff1}), we obtain
\begin{align}
\Pr\left[\big| Ld_{\max}\hat{ \DDelta}_{ij} \big|\ge t \right]&\le  2 \exp \left (  - \frac{2 ((t-\eps_L') \frac{2 \eta-1}{2 \eta+1})^2 }{  L }  \right ). \label{eqn:dev_Delta_ij}
\end{align}
Now by the same argument as in~\eqref{eqn:def_Delta_ii}, $Ld_{\max} \hat{\DDelta}_{ii} = -\sum_{k\ne i}Ld_{\max} \hat{ \DDelta}_{ik}= -\sum_{k\ne i} \hat{F}_{ik}$ so we have
\begin{equation}
|Ld_{\max}\hat{\DDelta}_{ii}|-\eps_L''\le |Ld_{\max}\hat{\DDelta}_{ii}|\le |Ld_{\max}\hat{\DDelta}_{ii}|+\eps_L''.
\end{equation}
 As a result, similarly to the calculation that led to \eqref{eqn:dev_Delta_ij}, we obtain
\begin{align}
\Pr\left[\big| Ld_{\max}\hat{ \DDelta}_{ii} \big|\ge t \right]&\le  2 \exp \left (  - \frac{2 ((t- \eps_L'') \frac{2 \eta-1}{2 \eta+1})^2 }{  Ld_{\max} }  \right ).\label{eqn:dev_Delta_ii}
\end{align}
 From the Hoeffding bound analysis leading to the non-asymptotic bound in~\eqref{eqn:dev_Delta_ii}, we know that by choosing
 \begin{equation}
 t := c\sqrt{Ld_{\max}\log n} \Big(\frac{2\eta+1}{2\eta-1}\Big)+\eps_L'', \label{eqn:choose_t}
 \end{equation}
 for some sufficiently large constant $c>0$,
 \begin{equation}
\Pr\left[\big| Ld_{\max}\hat{ \DDelta}_{ii} \big|\ge t \right] = O\Big( \frac{1}{\poly(n)} \Big).
 \end{equation}
 In other words,
\begin{equation}
|\hat{\DDelta}_{ii}|\lesssim \frac{1}{2\eta-1}\sqrt{\frac{\log n}{Ld_{\max}}} + \frac{\eps_L''}{Ld_{\max}} \label{eqn:bound_Delta}
\end{equation}
with  probability at least $1-1/\poly(n)$. Recall the definition of $\eps_L''$ in \eqref{eqn:def_epsLpp}. We now design $(\eps_L,\eps_L'')$ such that
\begin{equation}
\frac{\eps_L''}{Ld_{\max}}=\frac{16}{(2\eta-1)^2} \eps_L = \frac{1}{ 2\eta-1 } \sqrt[4]{ \frac{ \log^2 n}{L d_{\max}}} \label{eqn:choice_eps_L} .
\end{equation}
Now note   $d_{\max}=\Theta(\log n)$ with high probability. This implies that the second term in~\eqref{eqn:bound_Delta} dominates the first term.  
Thus,
\begin{equation}
|\hat{\DDelta}_{ii}|\lesssim    \frac{1}{ 2\eta-1 } \sqrt[4]{ \frac{ \log^2 n}{L d_{\max}}} ,\label{eqn:bound_Delta2}
\end{equation}
with  probability at least $1-1/\poly(n)$. A similar high probability bound, of course, holds for $|\hat{\DDelta}_{ij}|$ if we choose $t$ in \eqref{eqn:dev_Delta_ij} similarly to the choice made in \eqref{eqn:choose_t}. We may rearrange~\eqref{eqn:choice_eps_L} to yield
\begin{equation}
\eps_L  \asymp (2\eta-1)\sqrt[4]{ \frac{ \log^2 n}{Ld_{\max}}}. \label{eqn:eps_L_asymp}
\end{equation}
Given the bound on the diagonal elements $\hat{\DDelta}_{ii}$ in~\eqref{eqn:bound_Delta2} and a similar bound on the off-diagonal elements $\hat{\DDelta}_{ij}$, similarly to  the proof of Lemma~\ref{lemma:l2-norm-bound-etaknown} in Appendix \ref{app:ProofofLemmaBoundDelta}, the spectral norm of $\hat{\DDelta}$ can be bounded as
\begin{equation}
 \|\hat{\DDelta}\|\lesssim  \frac{1}{2\eta-1}\sqrt[4]{ \frac{  \log^2 n}{Ld_{\max}}}.
 \end{equation}
Now we check that the lower bound on $L$ is satisfied when we choose $\eps_L $ according to   \eqref{eqn:eps_L_asymp}. Using  the sample complexity bound in~\eqref{eqn:lower_bd_L}  and rearranging, we obtain
\begin{equation}
L\succsim\frac{\log n}{(2\eta-1)^4}
\end{equation}
which when combined with $S =  \binom{n}{2}pL$ is less stringent than the statement of Theorem \ref{thm:etaunknown}. This completes the proof of the upper bound of $\|\hat{\DDelta}\|$ in~\eqref{eqn:bd_upsilon}. 
 
 \subsection{Proof of Lemma \ref{lem:est_eta}}\label{app:prf_est_eta}
 
 Consider the functions $\lambda_1 : (1/2,1]\to\bbR$ and $\lambda_2 : (1/2,1]\to\bbR$  given by \eqref{eqn:gt}.
By direct differentiation, we have
\begin{equation}
\lambda_1'(t) =   \frac{-1}{(2t-1)^2 },\quad\mbox{and}\quad \lambda_2'(t) = \frac{-2}{(2t-1)^2 }.
\end{equation}
We note that an everywhere differentiable function $g$ is Lipschitz continuous with Lipschitz constant $\sup g'$. We now assume that $\eta,\hat{\eta} \in [\eta^*,1]$ for some $\eta^*>1/2$. By using the fact that $2 /(2\eta^*-1)^2$ is an upper bound of the derivative of $\lambda_j |_{[\eta^*,1]}$ (i.e., $\lambda_j$ restricted to the domain $[\eta^*,1]$),  one has
\begin{equation}
| \lambda_j(\hat{\eta} ) - \lambda_j(\eta) |\le\frac{2}{( 2\eta^*-1)^2} |\hat{\eta}-\eta|  \label{eqn:lip}
\end{equation}
for $j = 1,2$.  
We now put
\begin{equation}
\eta^* := \frac{1}{2}\Big(\eta+ \frac{1}{2}\Big).
\end{equation}
This quantity is the average of $1/2$ and $\eta$ and so is greater than $1/2$ as required. Also, $\eta-\eta^*=\eta/2-1/4$. Now,  \eqref{eqn:lip} becomes
\begin{equation}
|\lambda_j(\hat{\eta} ) - \lambda_j(\eta) |\le\frac{2}{( \eta-1/2)^2} |\hat{\eta}-\eta| = \frac{8}{(2\eta-1)^2}  |\hat{\eta}-\eta| \label{eqn:approx_error}
\end{equation}
for $j = 1,2$ if $\hat{\eta} \in [\eta^* , 2\eta-\eta^*] \subset [\eta^*,1] $.  The probability that this happens (recalling that $\hat{\eta}$ is the  random in question) is
\begin{align}
\Pr\big[\eta^*\le  \hat{\eta} \le 2\eta-\eta^*   \big]& = \Pr \bigg[ |\hat{\eta}-\eta| \le \frac{\eta}{2}-\frac{1}{4} \bigg] \\
& =1- \Pr \bigg[ |\hat{\eta}-\eta| > \frac{\eta}{2}-\frac{1}{4} \bigg].
\end{align}
From Lemma \ref{lem:fidelity}, we know that if
\begin{equation}
L \succsim \frac{1}{\eps^2}\log\frac{n}{\delta}, 
\end{equation}
then we have $|\hat{\eta}-\eta|\le\eps$ with probability at least $1-\delta$.  Hence, if
\begin{equation}
L\succsim\frac{1}{ (\frac{\eta}{2}-\frac{1}{4})^2} \log\frac{n}{\delta} \asymp \frac{1}{(2\eta-1)^2 }\log\frac{n}{\delta}
\end{equation}
then \eqref{eqn:approx_error} holds with probability at least $1-\delta$.  This completes the proof of Lemma \ref{lem:est_eta}. 

\section{Proof of Lemma~\ref{lem:scale}}\label{app:prf_scalings}
\subsection{The Scaling of Singular Values $\sigma_i(M_2)$} \label{sec:est_sing}
Since $M_2$ is  symmetric and positive semidefinite, its eigenvalues (which are all non-negative) are the same as its singular values. Since the eigenvectors are invariant to scaling, let us assume that
\begin{equation}
 v= \pi_0 + b \pi_1 \label{eqn:e_vector}
\end{equation}
is an eigenvector. Then by uniting the definition of $M_2$ in~\eqref{eqn:defM2} and \eqref{eqn:e_vector}, we have
\begin{equation}
M_2 v = (\eta \|\pi_0\|^2 + \eta b\langle \pi_0,\pi_1\rangle)\pi_0 + ((1-\eta )a \langle \pi_0,\pi_1\rangle+b(1-\eta)\|\pi_1\|^2 )\pi_1.
\end{equation}
Since $v$ is assumed to be an eigenvector, $M_2v$ satisfies that
\begin{equation}
M_2 v = \sigma v
\end{equation}
where $\sigma$ is some eigenvalue or singular value. Since $\pi_0$ is linearly independent of $\pi_1$, this equates to
\begin{align}
\eta \|\pi_0\|^2 + \eta b\langle \pi_0,\pi_1\rangle &= \sigma \label{eqn:simul1}\\
(1-\eta )a \langle \pi_0,\pi_1\rangle+b(1-\eta)\|\pi_1\|^2 &= \sigma b.\label{eqn:simul2}
\end{align}
Now note from the definitions of $\pi_0$ and $\pi_1$ that
\begin{equation}
\|\pi_0\|^2 =\|\pi_1\|^2
\end{equation}
because the elements are the same and $\pi_1$ is simply a permuted version of $\pi_0$.  So we will replace $\|\pi_1\|^2$ with $\|\pi_0\|^2 $ henceforth.
Eliminating $\sigma$ from the simultaneous equations in \eqref{eqn:simul1} and \eqref{eqn:simul2}, we obtain the quadratic equation in the unknown $b$:
\begin{align}
\eta \langle \pi_0,\pi_1\rangle  b^2 + (2\eta-1) \|\pi_0\|^2 b - (1-\eta) \langle \pi_0,\pi_1\rangle  = 0
\end{align}
which implies that
\begin{equation}
b^* = \frac{-(2\eta-1)\|\pi_0\|^2  \pm \sqrt{(2\eta-1)^2\|\pi_0\|^4+4\eta(1-\eta) \langle \pi_0,\pi_1\rangle^2  } }{2\eta\langle \pi_0,\pi_1\rangle  }. \label{eqn:b_star}
\end{equation}
 Now, we observe that
\begin{align}
\langle \pi_0,\pi_1\rangle  &= \sum_{(i,j) \in\calE} 2 \frac{w_iw_j}{w_i+w_j}\\
\|\pi_0\|^2 &=\sum_{(i,j)\in\calE}  \frac{w_i^2 +w_j^2}{(w_i+w_j)^2} .
\end{align}
so by the fact that $w_{\min}$ and $w_{\max}$ are bounded, we see that  $\langle \pi_0,\pi_1\rangle =\Theta(|\calE|)$ and $\|\pi_0\|^2=\Theta(|\calE|)$. Plugging these estimates into $b^*$, we see that $b^*=\Theta(1)$. Thus, by \eqref{eqn:simul1}, we see that  with high probability over the realization of the Erd\H{o}s-R\'enyi graph,
\begin{equation}
 \sigma = \Theta(\eta |\calE|) = \Theta\left( \eta n^2 p \right). \label{eqn:sigma_res}
 \end{equation}
 This scaling holds for both singular values $\sigma_1(M_2)$ and $\sigma_2(M_2)$ so this proves \eqref{eqn:sigma_res1}. Two distinct values for the singular values due to the $\pm$ sign in $b^*$ in~\eqref{eqn:b_star}.  This completes the proof of \eqref{eqn:sigma_res1}.

\subsection{The Scaling of Block-Incoherence Parameter $\mu(M_2)$} \label{sec:est_mu}
 Now let us evaluate the scaling of $\mu(M_2)$.   From \eqref{eqn:e_vector} and \eqref{eqn:b_star}, we know the form of the eigenvectors of $M_2$. The singular vectors must be normalized so they can be written as
 \begin{equation}
 \hatv :=\frac{v}{\|v\|_2}.
 \end{equation}
 Since the length of $v$ is $2|\calE|$, and the values (elements) of $v$ are uniformly upper and lower bounded, it is easy to see that $\|v\|_2 = \Theta( \sqrt{ |\calE|})$. As a result, one has
 \begin{equation}
 \hatv = \Theta \bigg( \frac{1}{\sqrt{|\calE|}}\bigg) v.
 \end{equation}
 Thus, each subblock of $U$ has entries that scale as $O(  |\calE|^{-1/2})$ and so
 \begin{equation}
  \big\| U^{(  k )} \big\|_2 = \Theta \bigg( \frac{1}{\sqrt{|\calE|}}\bigg) .
 \end{equation}
As a result, from the definition of $\mu(M_2)$ in \eqref{eqn:defUk},  we see that $\mu(M_2)$ is of constant order, i.e.,
\begin{equation}
 \mu(M_2) = \Theta(1),\label{eqn:mu_const0}
 \end{equation}
 which completes the proof of \eqref{eqn:mu_M2}.


\section{Bernstein inequality}
\label{app:BernsteinIneq}

\begin{lemma}
\label{lemma:BernsteinIneq}
Consider $n$ independent random variables $X_i$ with $|X_i| \leq B$. For any $\gamma \geq 2$, one has
\begin{align}
\left |  \sum_{i=1}^n X_i - \mathbb{E} \left[ \sum_{i=1}^n X_i \right] \right | \leq \sqrt{  2 \gamma \log n \sum_{i=1}^n \mathbb{E} \left[ X_i^2 \right] } + \frac{2 \gamma}{3} B \log n
\end{align}
with probability at least $1- 2n^{-\gamma}$.
\end{lemma} 

%
\bibliographystyle{ieeetr}
\bibliography{advtopKbib}

\end{document}